%% file: index.tex
\documentclass{ieeetj}
\usepackage{cite}
\usepackage{amsmath,amssymb,amsfonts}
\usepackage{algorithmic}
\usepackage{graphicx,color}
\usepackage{textcomp}
\usepackage{xcolor}
\usepackage{hyperref}
\hypersetup{hidelinks=true}
\usepackage{algorithm,algorithmic}
\def\BibTeX{{\rm B\kern-.05em{\sc i\kern-.025em b}\kern-.08em
    T\kern-.1667em\lower.7ex\hbox{E}\kern-.125emX}}
\AtBeginDocument{\definecolor{tmlcncolor}{cmyk}{0.93,0.59,0.15,0.02}\definecolor{NavyBlue}{RGB}{0,86,125}}

\usepackage{msc}

\usepackage{cite}
\usepackage{amsmath,amssymb,amsfonts,amsbsy}
\usepackage{algorithmic}
\usepackage{graphicx}
\usepackage{textcomp}
\usepackage{xcolor}
\usepackage{siunitx}
\usepackage{svg}
\usepackage{xurl}
\usepackage{hyperref}
\usepackage{amsthm}
\usepackage{mathtools}
\theoremstyle{plain}

\hypersetup{
    colorlinks=true,
    linkcolor=blue,
    urlcolor=blue,
    citecolor=blue,
}
\usepackage{tabularx}
\usepackage{subcaption}
\usepackage{bm}

\def\BibTeX{{\rm B\kern-.05em{\sc i\kern-.025em b}\kern-.08em
    T\kern-.1667em\lower.7ex\hbox{E}\kern-.125emX}}

\input{lib/proof}

\hypersetup{hidelinks=true}
\usepackage{algorithm,algorithmic}
\def\BibTeX{{\rm B\kern-.05em{\sc i\kern-.025em b}\kern-.08em
    T\kern-.1667em\lower.7ex\hbox{E}\kern-.125emX}}
\AtBeginDocument{\definecolor{tmlcncolor}{cmyk}{0.93,0.59,0.15,0.02}\definecolor{NavyBlue}{RGB}{0,86,125}}

\input{command}

\input{peformance}
\def\authorrefmark#1{\ensuremath{^{\textbf{#1}}}}

\def\OJlogo{}
\def\seclogo{}
\usepackage{etoolbox}

\makeatletter
\patchcmd{\@maketitle}
  {\ifx\@revised\@empty}
  {\ifx\@reviseddate\@empty}
  {}{}

\patchcmd{\@maketitle}
  {\fi.\space\ifx\@editor}
  {\fi\space\ifx\@editor}
  {}{}
\makeatother

\makeatletter

\let\origpsplain\ps@plain
\def\ps@plain{%
  \origpsplain
  \def\@oddfoot{\hbox to\textwidth{\hfill{\rffont\thepage}}}%
  \def\@evenfoot{\hbox to\textwidth{{\rffont\thepage}\hfill}}%
}

\let\origpsheadings\ps@headings
\def\ps@headings{%
  \origpsheadings
  \def\@oddfoot{\hbox to\textwidth{\hfill{\rffont\thepage}}}%
  \def\@evenfoot{\hbox to\textwidth{{\rffont\thepage}\hfill}}%
}

\makeatother
\begin{document}

\markboth{}{Akama {et al.}}

\title{STAR: Scalable TEE-based Anonymous Rate-limiter}

\author{Kosei Akama\authorrefmark{1,2}, Yoshimichi Nakatsuka\authorrefmark{1}, Koichi Moriyama\authorrefmark{2} and Keisuke Uehara\authorrefmark{2}}
\affil{National Institute of Information and Communications Technology, Koganei, Tokyo, Japan}
\affil{Keio University, Fujisawa, Kanagawa, Japan}
\corresp{Corresponding author: Kosei Akama (email: akama@nict.go.jp).}

\newcommand{\taggedpara}[1]{\noindent\textbf{#1.}}
\newcommand{\utaggedpara}[1]{\noindent\underline{#1.}}
\newcommand{\secretdepend}[1]{\textbf{\nolinkurl{#1}}}

\begin{abstract}

Protecting Internet services from attacks such as TLS/HTTP floods requires rate limiting, yet identifying users to enforce these limits can compromise their privacy.
Anonymous rate limiting promises to reconcile these goals, but existing schemes either struggle to scale or assume that services and authorities do not collude.
We present \sys{}, which uses server-side trusted execution environments (TEEs) to overcome these limitations. %
By enabling lightweight cryptographic operations, TEEs make authentication efficient, but introduce a privacy challenge: malicious services can observe memory-access patterns to infer user secrets.
We address this threat through countermeasures against secret-dependent memory accesses and assess our implementation using dynamic analysis.
To extend these efficiency gains across servers, we also design a secure sharding mechanism that distributes client requests.
We demonstrate that strong privacy guarantees can coexist with efficient authentication at scale; \sys{} processes authentication requests in tens of microseconds and handles approximately \totalThroughputNum{K} requests per second per CPU core.
\end{abstract}


\maketitle

\input{text/1_intro}

\input{text/2_background}
\input{text/3_related_work}
\input{text/4_system}
\input{text/5_design}
\input{text/6_protocol}

\input{text/7_implementation}

\input{text/8_evaluation}

\input{text/9_security}
\input{text/10_discussion}

\input{text/11_conlusion}

\bibliographystyle{IEEEtran}
\bibliography{citation}

\section*{Acknowledgment}
The authors thank Hayato Kimura (NICT) for advice on the security analysis of \sys. 

\appendix
\input{text/C_notation}
\input{text/D_proceedure}

\input{table/notation}
\input{text/B_side-channel-vector}
\input{text/A_proof}

\vfill\pagebreak

\end{document}

%% file: lib/proof.tex
\newtheorem{definition}{Definition}
\newtheorem{theorem}{Theorem}

\newcommand{\gametitlesep}{2pt}

\newenvironment{gameblock}[2][0.47\columnwidth]
{%
  \begin{minipage}[t]{#1}
  \vspace{0pt}%
  \centering
  {\scriptsize\bfseries #2}\par
  \vspace{\gametitlesep}%
  \hrule height 0.35pt
  \vspace{1pt}%
  \raggedright
  \scriptsize
  \begin{tabular}{@{}l@{}}
}
{%
  \end{tabular}
  \end{minipage}%
}

\newenvironment{queryblock}[1]
{%
  \centering
  {\scriptsize\bfseries #1}\par
  \vspace{1pt}%
  \hrule height 0.35pt
  \vspace{1pt}%
  \raggedright
  \scriptsize
  \begin{tabular}{@{}l@{}}
}
{%
  \end{tabular}\par
}

\newcommand{\gametoprule}{%
  \par
  \nointerlineskip
  \hrule height 0.45pt
  \vspace{\gametitlesep}%
}

\newcommand{\gamebottomrule}{%
  \vspace{2pt}%
  \hrule height 0.45pt
}

%% file: command.tex
\newcommand{\sys}{{\texttt{STAR}}}

\newcommand{\enc}{\mathtt{Seal}}
\newcommand{\dec}{\mathtt{Unseal}}
\newcommand{\mac}{\mathtt{MAC}}
\newcommand{\prf}{\mathtt{PRF}}
\newcommand{\hash}{\mathtt{Hash}}
\newcommand{\cred}{\mathtt{cred}}

\newcommand{\code}{\mathtt{code}}
\newcommand{\identity}{\mathtt{identity}}
\newcommand{\cnt}{\mathtt{cnt}}
\newcommand{\uid}{\mathtt{uid}}

\newcommand{\sk}{\mathtt{sk}}
\newcommand{\pk}{\mathtt{pk}}
\newcommand{\ck}{\mathtt{ck}}

\newcommand{\cipher}{\mathtt{cipher}}
\newcommand{\isequal}{\stackrel{?}{=}}

\newcommand{\usk}{\mathtt{usk}}
\newcommand{\gpk}{\mathtt{gpk}}
\newcommand{\IsMember}{\mathtt{IsMember}}

\newcommand{\test}{\mathtt{chall}}
\newcommand{\user}{\mathsf{user}}
\newcommand{\enclave}{\mathsf{enclave}}
\newcommand{\service}{\mathsf{service}}

\newcommand{\game}{\mathsf{Game}}

\newcommand{\setup}{\mathsf{Setup}}
\newcommand{\register}{\mathsf{Reg}}
\newcommand{\auth}{\mathsf{Auth}}

\newcommand{\unlink}{\mathsf{unlink}}

\newcommand{\ratelimit}{\mathsf{ratelimit}}
\newcommand{\oracle}{O}
\newcommand{\negl}{\mathsf{negl}}
\newcommand{\getsR}{\mathrel{\xleftarrow{\$}}}
\newcommand{\adversary}{\mathcal{A}}          %
\newcommand{\advantage}{\operatorname{Adv}} %
\newcommand{\bits}{\{0,1\}}

\newcommand{\leak}{\mathsf{Leak}}
\newcommand{\leakage}{\mathsf{leakage}}
\newcommand{\AccessPattern}{\mathtt{access}}
\newcommand{\AccessPatternLeakage}{\leakage_\AccessPattern}
\newcommand{\LeakAccessPattern}{\leak_\AccessPattern}

%% file: peformance.tex
\DeclareSIUnit{\cps}{cps}

\newcommand{\encryptLatency}{\SI{87.2}{\micro\second}}

\newcommand{\routeChooseLatency}{\SI{216}{\nano\second}}

\newcommand{\decryptLatency}{\SI{63.1}{\micro\second}}

\newcommand{\verificationLatency}{\SI{294}{\nano\second}}

\newcommand{\tokenGenerateLatency}{\SI{242}{\nano\second}}

\newcommand{\dedupeCheckLatency}{\SI{423}{\nano\second}}

\newcommand{\clientSumLatency}{\SI{87.4}{\micro\second}}

\newcommand{\serverSumLatency}{\SI{71.7}{\micro\second}}

\newcommand{\callLatency}{\SI{5.04}{\micro\second}}

\newcommand{\otherLatency}{\SI{2.20}{\micro\second}}

\newcommand{\scrappySignLatency}{\SI{1.36}{\milli\second}}

\newcommand{\scrappyVerifyLatency}{\SI{5.41}{\milli\second}}

\newcommand{\arcSignLatency}{\SI{1.62}{\milli\second}}

\newcommand{\arcVerifyLatency}{\SI{2.03}{\milli\second}}

\newcommand{\tlsLatency}{\SI{203}{\micro\second}}

\newcommand{\tlsMissExtensionLatency}{\SI{8.11}{\micro\second}}

\newcommand{\tlsInvalidCredLatency}{\SI{83.2}{\micro\second}}

\newcommand{\tlsReplayCredLatency}{\SI{83.6}{\micro\second}}

\newcommand{\payloadSize}{\SI{192}{\byte}}

\newcommand{\clientFasterThanExisting}{15.6}

\newcommand{\serverFasterThanExisting}{28.3}

\newcommand{\tlsThroughput}{\SI{4.92}{\kilo\cps}}

\newcommand{\tlsMissExtensionThroughput}{\SI{123}{\kilo\cps}}

\newcommand{\tlsInvalidCredThroughput}{\SI{12.0}{\kilo\cps}}

\newcommand{\tlsReplayCredThroughput}{\SI{12.0}{\kilo\cps}}

\newcommand{\throughputLargerThanTLS}{2.43}
\newcommand{\totalThroughputNum}{13.9}
\newcommand{\totalThroughput}{\SI{13.9}{\kilo\cps}}

\newcommand{\scrappyThroughput}{\SI{185}{\cps}}

\newcommand{\arcThroughput}{\SI{493}{\cps}}

\newcommand{\HLSize}{21,800}

\newcommand{\HLStorage}{\SI{1.08}{\mega\byte}}

%% file: text/1_intro.tex
\section{Introduction}

\subsection{Motivation}
In recent years, the abuse of Internet services has become increasingly problematic. In particular, attacks that generate far more requests than a service anticipates are a serious concern, as they waste server resources and can ultimately deny access to legitimate clients.
One way to address this issue is to identify users through identity verification and \emph{rate-limit} the number of requests a user can make.
However, this approach raises privacy concerns because it enables \emph{user tracking}, linking requests from the same user.
Consider the following cases where such a tension may occur:
\begin{itemize}
    \item \textbf{Web Application and API Protection (WAAP) } that block users who send excessive HTTP requests or establish excessive TLS connections, potentially as part of HTTP or TLS flooding attacks, may want to avoid using identifiers (e.g., cookies) for user privacy.
    
    \item \textbf{Prepaid services} that handle sensitive information (e.g., location information collected via toll roads) must limit each customer's usage to the purchased quota while preventing linking customer actions across sessions.
    
    \item \textbf{Services that use AI inference} that handle sensitive topics (e.g., personal health) must cap user requests to prevent large numbers of inference requests that lead to substantial API charges while protecting user privacy.
    
\end{itemize}

\subsection{Prior Work and Problem Statement}
To address this tension, previous work has proposed anonymous rate-limiting schemes.
Many of the prominent approaches rely on advanced cryptographic techniques (e.g., zero-knowledge proofs~\cite{kTAA,RLN,ARC,Scrappy}) or client-side TEEs~\cite{CACTI}. However, all of these prior works leverage cryptographic mechanisms that can impose high verification costs, which reduce a service's resilience to DoS and raise scalability challenges.
Other prior works introduce a third-party authority that is assumed not to collude with the service~\cite{PrivacyPass, Nymble} to perform rate limiting using a more lightweight cryptographic mechanism. However, there is no reliable way to guarantee non-collusion, and privacy concerns therefore remain.
This motivates an anonymous rate-limiting protocol that (1) is DoS resilient and scalable and (2) allows users to remain anonymous even if the service colludes with an authority.

\subsection{Our Approach}
In this paper, we propose \sys{} (\textbf{S}calable \textbf{T}EE-based \textbf{A}nonymous \textbf{R}ate-limiter), a novel anonymous rate-limiting scheme that utilizes a server-side Trusted Execution Environment (TEE) operated by the service to address the above challenges. First, by leveraging the confidentiality of a TEE, \sys{} protects user anonymity and remains robust against anonymity violations by an authority, even when the authority colludes with the service. Second, by executing security- and privacy-sensitive logic inside a TEE, \sys{} simplifies and accelerates processing, enabling high scalability by avoiding expensive cryptographic verification on the critical path.
However, leveraging a server-side TEE comes with its own challenges, as we discuss below.

\sys{} consists of two protocols: registration and authentication. In the registration protocol, the service verifies the user's identity and utilizes its TEE to issue a credential to the user. In the authentication protocol, the user encrypts the credential together with pieces of metadata (including an epoch and a counter) and submits the ciphertext to the server-side TEE. The TEE then decrypts the ciphertext, performs several checks, and, based on the decrypted values, computes a blinded one-time identifier, i.e., ``blind token''. This blind token is then sent to the TEE host (i.e., the service), to which the host checks whether it has seen this value in the past. If it has, the service denies the user access to its service, thus enforcing the rate-limit.
Most importantly, rate limiting is performed with respect to the blind token that is not linkable to the user's identity, realizing user privacy.

\taggedpara{Design Challenge \& Solution}
\sys{} achieves anonymity by using a TEE during authentication to conceal user-identifying information from the service itself. However, simply placing this functionality within the TEE is not enough, as we must take into account that a malicious service may compromise user anonymity by observing secret-dependent control flow or memory-access patterns, which may leak user-related secrets and enable user tracking.

To address this risk, we carefully design \sys{} so that it has no control-flow or memory-access patterns that depend on secrets. Specifically, when generating the blind token used for rate-limiting, \sys{} uses only algorithms without dangerous secret-dependent branches and memory accesses. 

Another challenge is regarding scaling \sys{} across multiple servers, that is, efficiently distributing the used-blind-token-sets. Distributed storage~\cite{stoica2001dht} can introduce performance overhead that may become a bottleneck for \sys{}, given its significantly lower processing cost than prior schemes.

To overcome this challenge, \sys{} employs sharding, where the destination server for each request is determined by a public value derived from its blind token. Each server can publicly verify the assignment, preventing malicious users from bypassing rate limits by replaying the same token across different shards. Moreover, because blind tokens are unlinkable, the routing information does not reveal user identities or enable requests from the same user to be linked.

\taggedpara{Advantages and Limitations}
\sys{} is designed to be practical in real-world environments. To demonstrate this, we implemented an end-to-end Proof-of-Concept (PoC) implementation consisting of a Core Library, a WAAP-like Rate-Limiting Service, a Web App Server, and a Client-side Proxy. 
We also use TLS extensions to carry \sys's payload, making \sys, to the best of our knowledge, the first anonymous rate-limiting implementation designed to mitigate TLS flooding.
Our open-sourced \sys{} code is publicly available, contributing to reproducibility\footnote{https://github.com/nict-sflab/star}.

To validate its scalability, we conducted a comprehensive performance evaluation of \sys{}. First, we measured server-side latency during authentication and demonstrate that \sys{} completes authentication processing on the order of a few hundred microseconds. This is about $\serverFasterThanExisting\times$ faster than the two prior works we compare against, namely ARC~\cite{ARC} and Scrappy~\cite{Scrappy}. Next, we measured server throughput and show that \sys{} can process approximately \totalThroughput{} (connections per second) authentication requests per CPU core. We also evaluated memory consumption, which is only \HLStorage{} per CPU core in Wikipedia-like large-scale settings. Finally, we evaluated communication overhead of \sys{} and this results only \payloadSize{} per authentication phase, which is smaller than that of prior work.

We also analyze the security of \sys{} by defining a formal security model using a game-based approach and proved that the \sys{} protocol satisfies the unlinkability and rate limitability properties, required for anonymous rate limiting, via hybrid games. %
In addition to the formal analysis, we performed a Timecop-style~\cite{Timecop} dynamic leakage analysis to search for secret-dependent branches and memory accesses in the Core Library and show that our implementation is free from access-pattern leakage.

Although \sys{} has numerous advantages, it also has limitations.
    Most notably, our access-pattern analysis focuses on secret-dependent control flow and memory accesses in the protocol logic; platform-dependent and physical leakage from the TEE is outside its scope.

\taggedpara{Contributions}
In summary, we introduce \sys{} that offers:
\begin{enumerate}
    \item a novel, scalable anonymous rate-limiting scheme built on a server-side TEE, with a design that avoids secret-dependent control flow and memory accesses to prevent leakage of user information, and sharding to distribute requests across multiple servers.
    \item practicality: our PoC confirms that \sys{} can be deployed under modeled real-world environment settings.
    \item high performance: server latency is on the order of a few hundred microseconds, throughput is \totalThroughput{} per CPU core, memory consumption is \HLStorage{} per CPU core, and communication cost is \payloadSize{}.
    \item security and anonymity: our dynamic analysis of secret-dependent branches and memory accesses and formal security analysis demonstrate that the \sys{} protocol satisfies rate-limitability and unlinkability.
\end{enumerate}

%% file: text/2_background.tex
\section{Background}

\subsection{Trusted Execution Environments}
A Trusted Execution Environment (TEE) is a hardware-backed execution environment that protects the confidentiality and integrity of its code and data, preventing unauthorized access even from privileged software (e.g., OS/Hypervisor).
Typically realized as specialized CPU features, this work assumes a TEE that provides the following fundamental guarantees:
\textbf{(i) Process Isolation}: A TEE isolates its program from all software running outside of it, protecting the program memory and internal state.
\textbf{(ii) Memory Encryption (ME)}: The internal state of a TEE-protected process is transparently encrypted using CPU-resident keys before being written to main memory. %
\textbf{(iii) Remote Attestation (RA)}: A TEE generates cryptographic evidence demonstrating to a remote verifier that (1) the communication endpoint is a genuine TEE backed by legitimate hardware, and (2) the exact intended program is executing within it.

Despite these robust protections, programs executing within TEEs remain vulnerable to adversaries inferring sensitive information by observing memory-access patterns (e.g., controlled channel attacks~\cite{ControlledSideChannelAttack}).

\subsection{$k$-Times Anonymous Authentication}
(Periodic) $k$-times anonymous authentication ($k$-TAA~\cite{kTAA}) is a privacy-preserving authentication scheme that enforces strict, per-period access limits.
Let $\usk$ denote the user's secret key, $\cred$ their membership credential, $\gpk$ the group public key, $\cnt$ the current usage counter, $K$ the maximum permitted counter value, and $t$ the current time window.
To authenticate, a user generates a zero-knowledge proof for the service, demonstrating the following statement (where $\prf(\cdots)$ denotes a PRF function and $\IsMember(\cdot, \cdot, \cdot)$ denotes a membership verification protocol):
\begin{align}
\textit{PK} \{ (\usk, \cnt) : h &= \prf(t, \cnt \parallel \usk) \land \cnt < K \notag \\
&\quad \land \IsMember(\usk, \cred, \gpk) = 1 \} \notag
\end{align}

The service accepts the authentication if the proof is successfully verified and the value $h$ has not been previously recorded for the current period $t$.

%% file: text/3_related_work.tex
\section{Related Work}
\label{sec:related-work}

To achieve both user anonymity and rate limiting, many anonymous rate-limiting schemes have been proposed in the literature, summarized in Table~\ref{tab:related-work}.
However, existing approaches generally suffer from at least one of the following two critical limitations: (1) they are susceptible to collusion among trusted authorities, which undermines user privacy, and/or (2) they exhibit poor scalability, where their deployment can inadvertently lead to denial-of-service (DoS) vulnerabilities on the services they are meant to protect.

\begin{table}[t]
    \centering
    \caption{Comparison of anonymous rate-limiting works. $T$ is the number of sub-periods in a period, $K$ is the max number that users can access the system, and $L$ is the number of pseudonyms used. 
    }
    \begin{tabular}{|c|c|c|}\hline
         Scheme & Anonymity Trust anchor & Server load \\ \hline \hline
         Privacy Pass & Non-colluding authorities & Light ($O(1)$) \\\hline
         Nymble & Non-colluding authorities & Light ($O(T)$) \\\hline
         k-TAA & None & Heavy ($O(K)$) \\\hline
         RLN & None & Heavy ($O(\log K + \log L) $) \\\hline
         ARC & None & Heavy ($O(\log K)$) \\\hline
         Scrappy & None & Heavy ($O(1)$) \\\hline
         CACTI & None & Heavy ($O(1)$) \\\hline
         \sys{} & \textbf{TEE} & \textbf{Light ($O(1)$)} \\\hline
    \end{tabular}
    \label{tab:related-work}
\end{table}

The first category of approaches relies on the assumption of non-colluding authorities, namely Privacy Pass~\cite{PrivacyPass} and Nymble~\cite{Nymble}.
Privacy Pass utilizes cryptographic techniques to restrict access frequency while preserving user anonymity. Nevertheless, this system remains vulnerable to timing-correlation attacks if the underlying authorities collude.
Although Nymble is an anonymous blocklisting system, it also limits user access rate.
However, if authorities in these systems collude with the service, users can be completely de-anonymized.
Furthermore, due to the inherent design constraints of anonymous blocklisting, Nymble strictly prohibits a user from maintaining multiple pseudonyms within a single time window, reducing flexibility.

Another category of solutions, including (periodic) $k$-times anonymous authentication ($k$-TAA)~\cite{kTAA} and related protocols like Scrappy~\cite{Scrappy}, ARC~\cite{ARC}, and RLN~\cite{RLN}, leverages zero-knowledge proofs (ZKPs) to mathematically guarantee that a user's access frequency remains below a predefined threshold without revealing any information to the service.
However, generating and verifying ZKPs incurs substantial computational overhead.
Under high-frequency access scenarios, such as those generated by automated bots, this overhead strictly bounds the service's overall scalability.

Alternatively, CACTI~\cite{CACTI} explored constructing anonymous rate limiting by enforcing access limits using client-side TEEs. However, a fundamental limitation of this approach is that it excludes clients operating on hardware that does not support TEEs.
Furthermore, to preserve user anonymity, CACTI requires the service to verify expensive privacy-preserving TEE attestation mechanisms (in CACTI's case, EPID~\cite{EPID}), which introduces scalability bottlenecks.

%% file: text/4_system.tex
\section{System and Threat Model}
\label{sec:system}
\subsection{Intuitive Definition}
This section outlines the system architecture and the threat model considered in this work as well as the requirements that must be satisfied by \sys{}.

The primary entities within our system are:
\begin{itemize}
    \item \textbf{User:} that consumes the provided service.
    \item \textbf{Service:} that provides services to users.
    \item \textbf{Auth Enclave:} which is a trusted component responsible for authenticating users.
    The enclave logic strictly executes within a server-operated TEE.
\end{itemize}

We consider the user and service to be potentially malicious.
A malicious user attempts to attack the service by exceeding the maximum allowed threshold of $K$ access attempts within a given time window.
The goal of a malicious service is to invade user privacy; specifically, the service may log and analyze all available user requests and memory access traces as well as modify the protocol ordering in an attempt to track or de-anonymize users.

Under these anticipated threats, \sys{} must be designed to satisfy the following core security requirements:
\begin{itemize}
    \item \textbf{Rate-Limiting:} A user can successfully authenticate at most $K$ times per time window; and
    \item \textbf{Unlinkability:} The service cannot track a user across multiple authentication executions, i.e., the service cannot ascertain whether any two cryptographic payloads and observable memory-access traces originated from the same user.
\end{itemize}

In addition to these security requirements, \textbf{scalability} is a critical non-security requirement.
Consequently, both the Auth Enclave and the underlying service in \sys{} must be highly performant and capable of handling a massive volume of concurrent requests from a large user base.

To achieve these requirements, we establish several security assumptions:
First, we exclude adversaries that are typically considered out of scope in TEE-based systems, including adversaries that directly compromise the security guarantees of the TEE platform itself.
    For example, we assume an OS/hypervisor-level adversary but strictly exclude hardware-level adversaries.
        Specifically, we consider leakage caused by secret-dependent control flow and memory-access patterns, including array accesses, that can be observed by the service~\cite{ControlledSideChannelAttack}.
        Platform- or microarchitecture-dependent attacks whose success is independent of the protocol logic are outside of this work's scope.
Second, \sys{} adopts several assumptions common to other rate-limiting schemes.
    In particular, \sys{} assumes that users have already correctly obtained the global parameters shared by all entities, such as the TEE public key $\pk$.
    To this end, we can consider that these global parameters are published on a trusted bulletin board (e.g., Certificate Transparency logs) together with \sys{}'s RA evidence generated by the TEE.
    Using these materials, users are assumed to have verified that the global parameters are authentic, are bound to a genuine TEE, and correspond to the expected \sys{} code before executing the protocol.
    Additionally, we treat clock skew between clients as orthogonal to our core focus and assume that all clients synchronize their clocks to the correct current time.
    We also exclude network DoS attacks operating below Layer 5 (e.g., TCP SYN floods) and IP-based tracking attacks.
    We assume that these threats are adequately mitigated through domain-specific deployment mechanisms, such as L3 firewalls for lower-layer DoS attacks and Tor for IP-based tracking.
    Finally, Sybil attacks are similarly orthogonal to our core focus.
    Consequently, we assume that \sys{} integrates robust identity-verification mechanisms during registration, such as the hardware-backed certificates used in \cite{Scrappy}, so that a single user cannot register multiple distinct accounts.

\subsection{Formal Definition}
\label{sec:formal-def}
The \sys{} and the security requirements of \sys{} can be defined formally as follows:

\begin{definition}[The definition of \sys{}]
\sys{} is a set of the following processes:

\begin{itemize}
    \item {\bf $\register = (\register_{\user, 1}, \register_{\enclave}, \register_{\user, 2})$:}
    This process allows a new user to register with the \sys{} system. First, the user runs $(\cipher, \ck) \leftarrow \register_{\user, 1}()$ and sends cipher $\cipher$ to the service. The service then invokes the TEE to compute $\cipher' \leftarrow \register_{\enclave}(\cipher)$ and returns $\cipher'$ to the user. Finally, the user retrieves the credential by running $\cred \leftarrow \register_{\user, 2}(\cipher', \ck)$ with the common key $\ck$.
    \item {\bf $\auth = (\auth_{\user}, \auth_{\enclave}, \auth_{\service})$:} 
    This process allows users to send requests to the service subject to the rate limit. First, the user generates $(\cipher'', s, \cnt) \leftarrow \auth_\user(\cred, \cnt, \pk)$ with a user counter $\cnt$ and sends the resulting request $\cipher''$ to the service indexed $s$. The service then invokes the TEE to compute $h \leftarrow \auth_\enclave(\cipher'', t', \sk)$. Finally, the service verifies the TEE's response and then inserts $h$ into duplication checking list by running $HL_s \leftarrow \auth_\service(h, HL_s)$.
\end{itemize}

\end{definition}

\input{image/rate-limiting}

\input{image/unlinkability.tex}
\begin{definition}[Rate-Limitability of the \sys{}]
Rate-Limitability is defined as the property that, for any probabilistic polynomial-time (PPT) adversary $\adversary$, the probability that $\adversary$ outputs $1$ in the rate-limiting game shown in Figure~\ref{fig:rate-limitability} is negligible in the security parameter $\lambda$; that is,

$$
    \Pr\!\left[
        \game_\adversary^\ratelimit(1^\lambda)=1
    \right]
    \leq \negl(\lambda).
$$

\end{definition}

It is noteworthy that this rate-limitability game models the worst-case setting in which only adversarial users are considered. Accordingly, the adversary wins if it achieves more than $K \times N$ successful authentications within a single time window, where $N$ denotes the number of malicious users registered by the adversary.

\begin{definition}[Unlinkability of the \sys{}]
For any PPT adversary $\adversary$, we define its advantage in the Unlinkability game shown in Figure~\ref{fig:unlinkability} as

$$
\advantage^\unlink_{\adversary}(\lambda)
\coloneqq
\left|
\Pr\!\left[
\game^\unlink_{\adversary}(1^\lambda)=1
\right]
-\frac{1}{2}
\right|.
$$

The protocol satisfies unlinkability if, for every PPT adversary $\adversary$, there exists a negligible function $\negl$ such that

$$
\advantage_\adversary^\unlink(\lambda)
\leq
\negl(\lambda).
$$
\end{definition}

Notably, in the unlinkability game, the adversary can observe the Enclave's memory-access patterns during both registration and authentication, denoted by $\AccessPatternLeakage$. Thus, \sys{} must ensure that this leakage does not enable the adversary to determine whether two authentication requests originate from the same user.

%% file: image/rate-limiting.tex
\newcommand{\refresh}{\mathtt{Refresh}}
\begin{figure}[t]
  \centering
  \gametoprule

  \begin{gameblock}{$\game_{\adversary}^{\ratelimit}(\lambda)$}
    $(\pk,\sk, HL_s) \gets \setup(1^\lambda)$\\
    $t' \leftarrow 0$\\
    $i \leftarrow 0$\\
    $N \leftarrow 0$\\
    $1 \gets \adversary^{\oracle_\ratelimit}()$\\[2pt]

    \textbf{return} $[i > K \times N]$
  \end{gameblock}%
  \hfill
  \begin{minipage}[t]{0.47\columnwidth}
    \vspace{0pt}

    \begin{queryblock}{$\oracle_\ratelimit^{\register}(\cipher)$}
      $\cipher' \gets \register_\enclave(\cipher, \sk)$\\
      $N \leftarrow N + 1$\\
      \textbf{return} $\cipher'$
    \end{queryblock}

    \vspace{6pt}

    \begin{queryblock}{$\oracle_\ratelimit^{\auth}(\cipher'', s)$}
      $h \gets \auth_\enclave(\cipher'', t', \sk)$\\
      $HL_s \gets \auth_\service(h,HL_s)$\\
      $i \leftarrow i + 1$\\
      \textbf{return} $1$
    \end{queryblock}
    
    \vspace{6pt}

    \begin{queryblock}{$\oracle_\ratelimit^{\refresh}()$}
      $t' \leftarrow t' + 1$\\
      $i \leftarrow 0$\\
      $HL_s \leftarrow \emptyset$\\
      \textbf{return} $1$
    \end{queryblock}
  \end{minipage}

  \gamebottomrule
  \caption{Rate-limitability security game.}
  \footnotesize $\oracle_\ratelimit = \{\oracle_\ratelimit^\register,
                    \oracle_\ratelimit^\auth, \oracle_\ratelimit^\refresh\}$ \\
  All oracles stop and return $\bot$ when it fails (e.g., assertion fails).

  \label{fig:rate-limitability}
\end{figure}

%% file: image/unlinkability.tex
\begin{figure}[t]
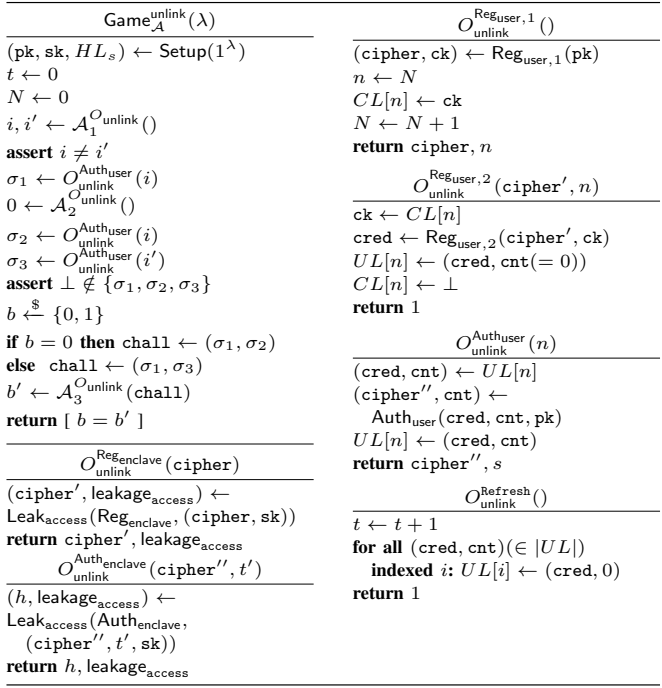

  \centering
  \gametoprule

  \begin{minipage}[t]{0.47\columnwidth}
    \vspace{0pt}

    \begin{gameblock}[\linewidth]{$\game_{\adversary}^{\unlink}(\lambda)$}
      $(\pk,\sk, HL_s) \gets \setup(1^\lambda)$\\
      $t \leftarrow 0$\\
      $N \gets 0$ \\
      $i, i' \gets \adversary_1^{\oracle_\unlink}()$\\[2pt]
      \textbf{assert} $i \neq i'$ \\
      $\sigma_1 \gets \oracle_\unlink^{\auth_\user}(i)$ \\
      $0 \leftarrow \adversary_2^{\oracle_\unlink}()$\\[2pt]

      $\sigma_2 \gets\oracle_\unlink^{\auth_\user}(i)$ \\
      $\sigma_3 \gets \oracle_\unlink^{\auth_\user}(i')$ \\
      \textbf{assert} $\bot \notin \{\sigma_1, \sigma_2, \sigma_3 \}$ \\
      $b \getsR \bits$\\[2pt]
      \textbf{if $b = 0$ then} $\test \gets (\sigma_1, \sigma_2)$ \\
      \textbf{else } $\test \gets (\sigma_1, \sigma_3)$ \\
      $b' \gets \adversary_3^{\oracle_\unlink}(\test)$\\[2pt]
      \textbf{return} [ $b=b'$ ]
    \end{gameblock}

    \vspace{4pt}
    \hrule height 0.35pt
    \vspace{2pt}

    \begin{queryblock}{$\oracle_\unlink^{\register_\enclave}(\cipher)$}
      $(\cipher', \AccessPatternLeakage) \gets$\\
      $\LeakAccessPattern(\register_\enclave, (\cipher, \sk))$\\
      \textbf{return} $\cipher', \AccessPatternLeakage$
    \end{queryblock}

    \vspace{1pt}

    \begin{queryblock}{$\oracle_\unlink^{\auth_\enclave}(\cipher'', t')$}
      $(h, \AccessPatternLeakage) \gets$\\
      $\LeakAccessPattern(\auth_\enclave,$\\
      $\quad (\cipher'', t', \sk))$\\
      \textbf{return} $h, \AccessPatternLeakage$
    \end{queryblock}
  \end{minipage}%
  \hfill
  \begin{minipage}[t]{0.47\columnwidth}
    \vspace{0pt}

    \begin{queryblock}{$\oracle_\unlink^{\register_{\user, 1}}()$}
      $(\cipher, \ck) \gets \register_{\user, 1}(\pk)$\\
      $n \gets N$\\
      $CL[n] \gets \ck$\\
      $N \gets N + 1$\\
      \textbf{return} $\cipher, n$
    \end{queryblock}

    \vspace{4pt}

    \begin{queryblock}{$\oracle_\unlink^{\register_{\user, 2}}(\cipher', n)$}
      $\ck \gets CL[n]$\\
      $\cred \gets \register_{\user,2}(\cipher', \ck)$\\
      $UL[n] \gets (\cred, \cnt(=0))$\\
      $CL[n] \gets \bot$\\
      \textbf{return} $1$
    \end{queryblock}

    \vspace{4pt}

    \begin{queryblock}{$\oracle_\unlink^{\auth_\user}(n)$}
      $(\cred, \cnt) \gets UL[n]$\\
      $(\cipher'', \cnt) \gets$\\
      $\quad \auth_\user(\cred, \cnt, \pk)$\\
      $UL[n] \gets (\cred, \cnt)$\\
      \textbf{return} $\cipher'', s$
    \end{queryblock}

    \vspace{4pt}

    \begin{queryblock}{$\oracle_\unlink^{\refresh}()$}
      $t \leftarrow t + 1$\\
      \textbf{for all $(\cred, \cnt ) (\in |UL|)$} \\
      $\quad$\textbf{indexed $i$:} 
      $UL[i] \gets (\cred, 0)$\\
      \textbf{return} $1$
    \end{queryblock}
  \end{minipage}

  \gamebottomrule
  \caption{Unlinkability security game.}
  \footnotesize
    \footnotesize $\oracle_\unlink = \{
    \oracle_\unlink^{\register_{\user,1}},
  \oracle_\unlink^{\register_{\user,2}},
  \oracle_\unlink^{\register_\enclave},
  \oracle_\unlink^{\auth_\user},
  \oracle_\unlink^{\auth_\enclave}
    \}$ \\
  All oracles stop and return $\bot$ when it fails (e.g., assertion fails).
    $(y, \leakage) \leftarrow \LeakAccessPattern(f, x)$ is a function resulting $y (\leftarrow f(x))$ with memory access pattern of $f$ $\leakage$.
  \label{fig:unlinkability}
\end{figure}

%% file: text/5_design.tex
\section{Design Challenges}

\subsection{Rate-Limiting with Oblivious Access Pattern}
\label{sec:access_side_channel}

\taggedpara{Challenge}
Whether to accept a request depends on the user's past access rate.
A straightforward approach to determine this is to maintain a per-user access counter inside the server-side TEE and consult it during authentication.
However, by observing which per-user record is accessed by the TEE, the service may potentially compromise unlinkability.

Common defenses against this attack, notably linear scan~\cite{Signal} and ORAM~\cite{ORAM}, have limitations.
A linear scan must traverse the entire state on every access to hide which entry was accessed, resulting in $O(M)$ time per request and poor performance at scale.
ORAM hides access patterns by randomizing storage locations, but it incurs high latency of $O(\log^2 M)$, where $M$ is the number of entries~\cite{PathORAM, Oblix, EnigMap}. Moreover, it complicates parallelization, leading to additional scalability issues~\cite{SONIC, DUORAM}.%

Conducting per-user state lookup during authentication is not ideal.
We must, therefore, explore an alternative mechanism that satisfies several requirements simultaneously.
First, because a per-user identifier is necessary for such a mechanism to enforce rate-limiting, its value must not be exposed through control flow or memory-access patterns.
Second, the mechanism must enforce that a user can authenticate at most $K$ times within a time window.
Third, distinct authentication attempts by the same user should remain unlinkable, except when the user reuses the same rate-limiting slot, which must be detected for rate-limiting purposes.
Finally, the mechanism should be efficient and compatible with parallelization.
Satisfying all of these requirements within a server-side TEE is nontrivial.

\taggedpara{Solution}
We address this challenge by introducing a mechanism that leverages a blinded per-request identifier, which we refer to as a \emph{blind token}, shown in Figure~\ref{fig:design}.
For each authentication attempt, the Auth Enclave computes a blind token $h$ using a sufficiently long user ID $\uid$ known only to the user and the enclave in addition to a counter $\cnt \in \{0,\ldots,K-1\}$ and the current time window $t$, both supplied by the user:
\[
    h \leftarrow \hash(t, \cnt, \uid).
\]
Here, $\cnt$ and $\uid$ serve as the secret key.
Similar constructions are used in $k$-TAA~\cite{kTAA}, ARC~\cite{ARC}, Nymble~\cite{Nymble}, and APHE~\cite{PHE}.

\begin{figure}
    \centering
    \footnotesize
    \includegraphics[width=0.8\linewidth]{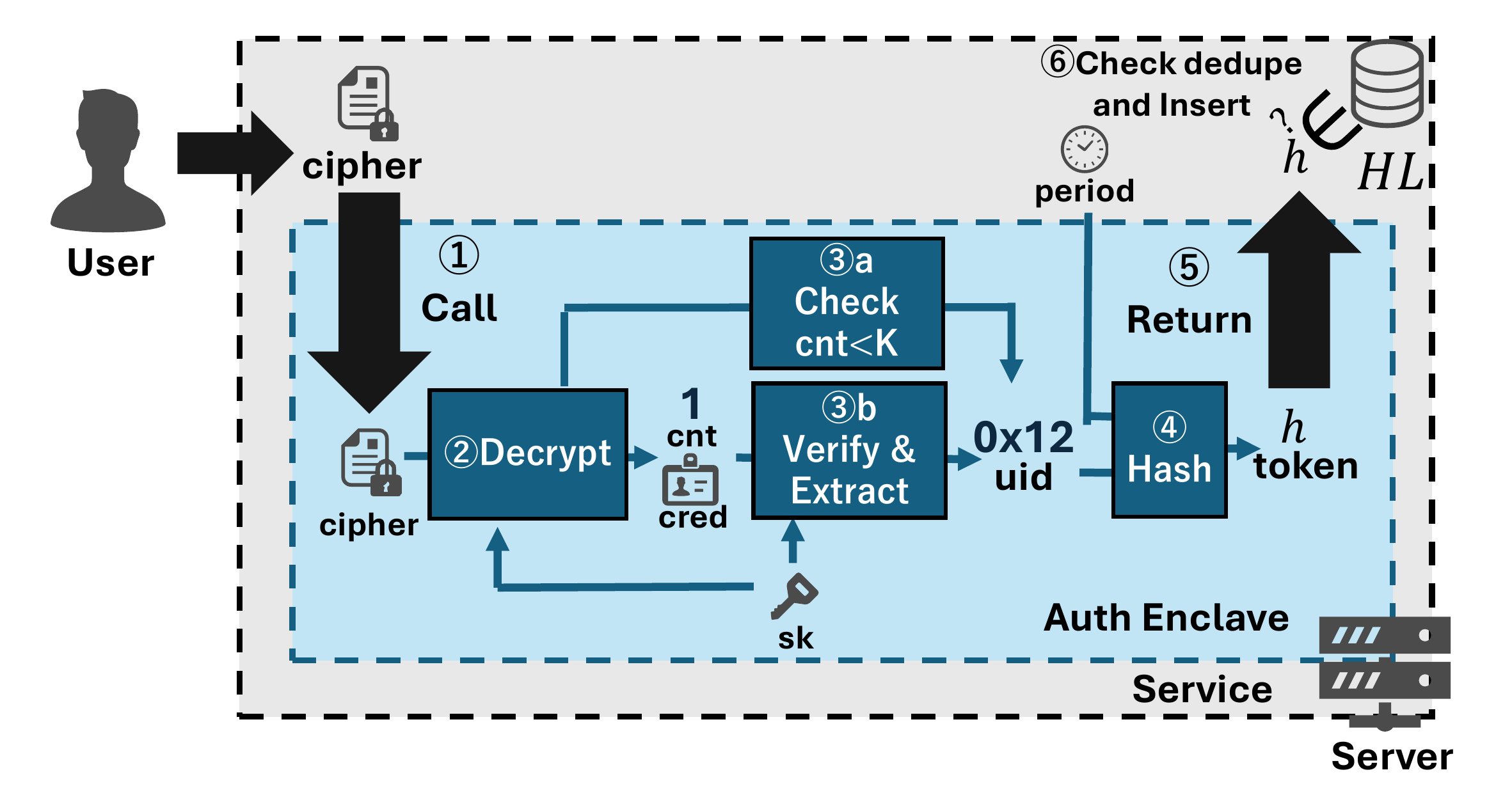}
    \caption{Blind-token mechanism of \sys{}}
    \label{fig:design}
\end{figure}

The Auth Enclave first verifies that the user-supplied counter satisfies $\cnt < K$.
A user may reuse a previously used $\cnt$ within the same time window; however, this produces the same blind token that has been redeemed before and is therefore rejected by the service.
Any $\cnt$ values outside the permitted range are rejected by the enclave.
Thus, at most $K$ distinct blind tokens can be successfully redeemed by a user within each time window.
Once this check passes, the Auth Enclave generates $h$.

After generating $h$, the Auth Enclave returns it to the service, which checks whether the token has already been redeemed within the current time window.
Specifically, the service maintains a list of already-redeemed blind tokens and performs duplicate detection over this list.
The use of blind token and delegation of rate-limiting to the service introduces the following two advantages.
First, because $h$ is calculated using a one-way function, the service cannot recover the user identifier from $h$.
Second, because the enclave does not perform rate limiting itself, it does not access user-specific state, removing the need for access-pattern-hiding approaches.
    This allows us to employ efficient $O(1)$ look-up data structures such as hash sets.

Because $h$ is pseudorandom and varies with $\cnt$ and $t$, authentication attempts by the same user are unlinkable across distinct $(\cnt,t)$ pairs.
Moreover, we implement blind-token generation and the $\cnt$ range check with secret-independent branches and memory-access patterns.
Enclave-side processing is limited to credential verification, counter checking, and token generation.
This keeps the enclave codebase small and makes the access-pattern analysis in Section~\ref{sec:security} manageable.

There are several works that leverage blind tokens (or similar approaches), although their techniques cannot be applied directly to \sys{} due to limitations in functionality or performance.
For example, APHE~\cite{PHE} uses a blinded value to limit login attempts without exposing the account name. However, it does not enforce a per-user limit on service access \emph{after} a successful login, which is \sys{'s} target. %
ZIPNet~\cite{ZIPNet} utilizes blinded messages to construct an anonymous broadcasting system.
    Although ZIPNet also introduces some form of rate-limiting functionality, applying their approach to the \sys{} use-case (e.g., WAAP) would introduce extremely high bandwidth consumption since they require all users to transmit a real or cover message for anonymity purposes.
Blinded tokens leveraged in k-TAA~\cite{kTAA} and ARC~\cite{ARC} are computed on the client side.
    Due to this, clients must prove possession of the correct identifier without revealing the identifier itself, which requires computationally expensive zero-knowledge proofs.
    This also requires the algorithm to use hash functions optimized for zero-knowledge proofs, which is slower compared to general-purpose CPU-optimized ones such as SHA-256, which is used by \sys{}.

\subsection{Distributing $HL$ for Scalability}
\label{sec:sharding}

\taggedpara{Challenge}
In the previous sections, we described how we can increase the scalability of \sys{} by avoiding expensive cryptographic algorithms and data structures that hide access patterns.
To further increase its scalability, \sys{} should be deployed in a distributed setting.
    There are two key information that must be distributed; $\sk$ stored within the Auth Enclaves (discussed in Section~\ref{sec:impl:key-share}) and $HL$.
One way of maintaining $HL$ in a distributed setting is to run \sys{} on top of a distributed database (e.g., distributed hash table~\cite{stoica2001dht}) to leverage its atomicity guarantee and prevent race conditions.
However, this approach can incur performance bottlenecks due to mutex locking, data merging, and data transfer in the service's intranet.
Unlike prior work in which cryptographic processing is the primary bottleneck, this issue is especially significant in \sys{}, where per-request processing is expected to complete in the order of microseconds.

\taggedpara{Solution}
To address this issue, we employ the \sys{} protocol with partition $HL_s$ on a per-server basis (i.e., sharding).
Specifically, when the user computes the blind token $h$ from $\uid (\in \cred)$, it also computes $s$, which is the remainder obtained after dividing $h$ by the number of servers $S$ and routes the request to the server corresponding to $s$.
With this approach, each server maintains only the subset of $h$ values, corresponding to approximately $1/S$ of $HL$, and, therefore, user requests are distributed approximately evenly at a rate of $1/S$, assuming that $h$ is uniformly distributed.

Several security considerations must be taken into account during the design of this solution.
First, a malicious user may send a request to a server that does not correspond to $s$. %
    However, the server can always reject this request by recomputing $s$ and verifying whether it is assigned to that shard.
Second, a malicious server may abuse this feature to track users via the routing value.
    Such tracking is computationally infeasible because the current $s$ does not provide any leverage to a malicious server in predicting the next $s$, since $h$ is computed using a Hash.
Finally, a malicious user may attempt to take down a specific server by selecting invalid routes and repeatedly accessing a single server.
    However, the impact of this attack on \sys{} is limited. This is because, owing to the sharding approach, even if one server becomes unavailable due to the above attack, the number of affected user requests is limited to only $1/S$ of the total number of requests, assuming $S$ is sufficiently large.
    Moreover, even if a user is affected, the impact remains limited, because the user can always retransmit the request using another counter, as long as $K$ is sufficiently large and the epoch duration is short.

\begin{figure}
    \centering
    \footnotesize
    \includegraphics[width=0.75\linewidth]{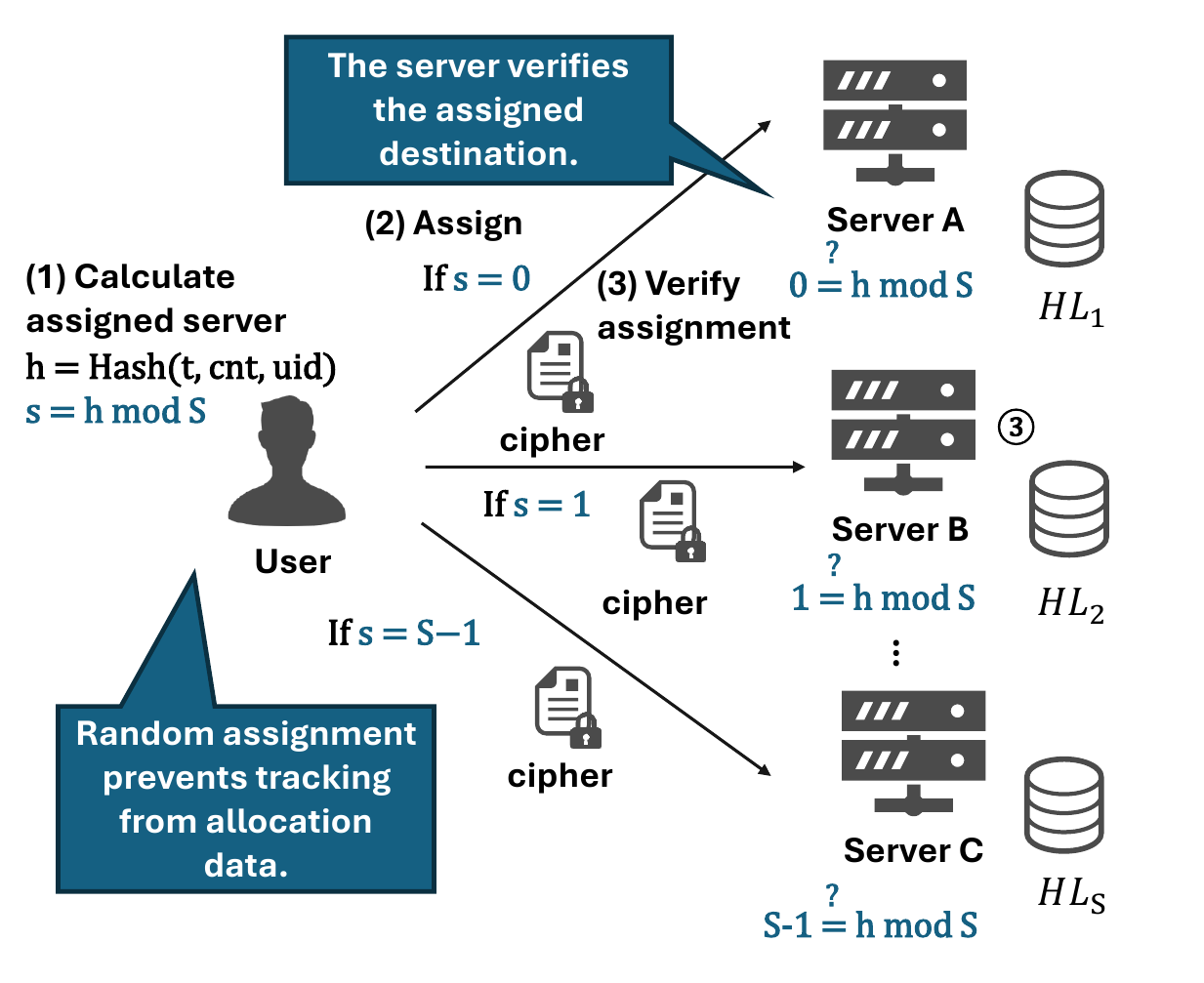}
    \caption{The sharding mechanism of \sys}
    \label{fig:sharding}
\end{figure}

%% file: text/6_protocol.tex
\section{The \sys{} protocol}
\label{sec:protocol}

This section details the \sys{} protocol.
To clearly present the steps of the protocol, we utilize the notation defined in Table~\ref{tab:notation} (in Appendix~\ref{sec:notation}) and show the protocol in Figures~\ref{fig:registration} and \ref{fig:authentication}.
For the underlying cryptographic operations, we assume that all variables are securely padded to fixed lengths, and that the secret key of the enclave $\sk$, and the common key between users and the enclave $\ck$ and the identifiers $\uid$ are of sufficient length (e.g., 256 \, bits) to prevent brute-force attacks.
Furthermore, we let $\mac$ denote a cryptographic MAC, $\hash(\cdot)$ a hash function, $\enc$ and $\dec$ AEAD sealing and unsealing, respectively, and $\enc'$ and $\dec'$ public-key encryption and decryption, respectively.
Additionally, $\pk$ represents the public key of the Auth Enclave, which the client is assumed to have verified a priori via remote attestation.
Moreover, although we use $\sk$ to denote all secret keys of the Auth Enclave, the keys for HMAC and decryption should be distinct in the implementation.
Finally, after each transition of the time window, the user's session counter $\cnt$ is reset to $0$ and the duplicate-check list $HL_s$ is flushed ($HL_s = \emptyset$).

\input{image/register}
\input{image/auth}

%% file: image/register.tex
\begin{figure}[t]
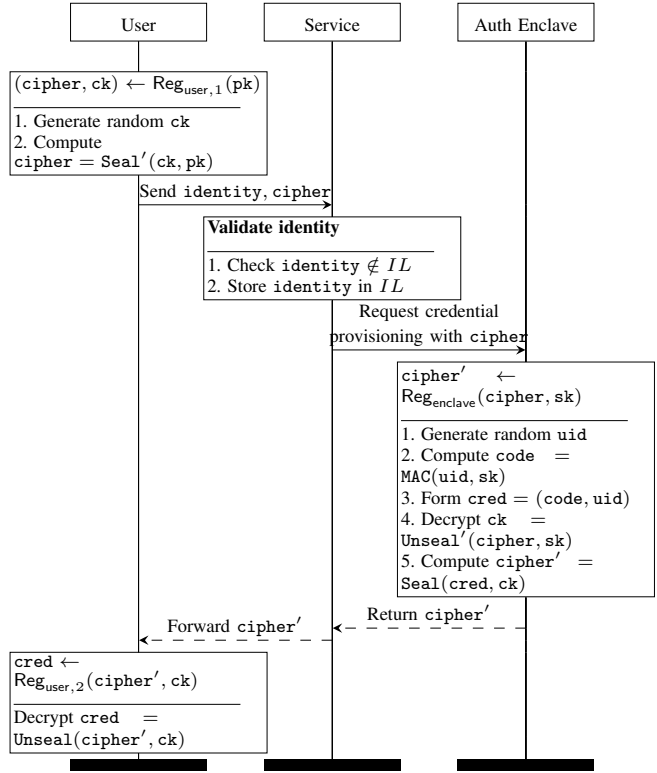

    \centering
    \begin{msc}[
        small values,
        draw frame=none,
        msc keyword=,
        head top distance=4mm,
        foot distance=0mm,
        instance distance=7.5mm,
        environment distance=8mm,
        level height=1.6mm,
        arrow scale=1.6,
        instance width=18mm,
        action width=34mm,
        action height=3mm,
        left inline overlap=11mm,
        right inline overlap=11mm,
        /tikz/font=\scriptsize,
        /tikz/line width=0.35pt,
        /msc/line width=0.35pt,
        label distance=0.3ex,
        every action/.style={align=left}
    ]{}

        \declinst{u}{}{User}
        \declinst{s}{}{Service}
        \declinst{ae}{}{Auth Enclave}

        \action{
            $(\cipher,\ck) \gets \register_{\user,1}(\pk)$\\[-0.3ex]
            \rule{30mm}{0.3pt}\\[-0.3ex]
            1.~Generate random $\ck$\\
            2.~Compute\\[-0.3ex]
            $\cipher = \enc'(\ck,\pk)$
        }{u}

        \nextlevel[11]

        \mess{Send $\identity,\cipher$}{u}{s}

        \nextlevel[1]

        \action{
            \textbf{Validate identity}\\[-0.3ex]
            \rule{30mm}{0.3pt}\\[-0.3ex]
            1.~Check $\identity \notin IL$\\
            2.~Store $\identity$ in $IL$
        }{s}

        \nextlevel[11]

        \mess{\shortstack{Request credential\\provisioning with $\cipher$}}{s}{ae}

        \nextlevel[1]

        \action{
            $\cipher' \gets \register_\enclave(\cipher,\sk)$\\[-0.3ex]
            \rule{30mm}{0.3pt}\\[-0.3ex]
            1.~Generate random $\uid$\\
            2.~Compute $\code = \mac(\uid,\sk)$\\
            3.~Form $\cred = (\code,\uid)$\\
            4.~Decrypt $\ck = \dec'(\cipher,\sk)$\\
            5.~Compute $\cipher' = \enc(\cred,\ck)$
        }{ae}

        \nextlevel[22]

        \mess*{Return $\cipher'$}{ae}{s}

        \nextlevel[1]

        \mess*{Forward $\cipher'$}{s}{u}

        \nextlevel[1]

        \action{
            $\cred \gets$\\[-0.3ex]
            $\register_{\user,2}(\cipher',\ck)$\\[-0.3ex]
            \rule{30mm}{0.3pt}\\[-0.3ex]
            Decrypt $\cred = \dec(\cipher',\ck)$
        }{u}

        \nextlevel[7]

    \end{msc}

    \caption{Registration protocol.}
    \label{fig:registration}
\end{figure}

%% file: image/auth.tex
\begin{figure}[t]
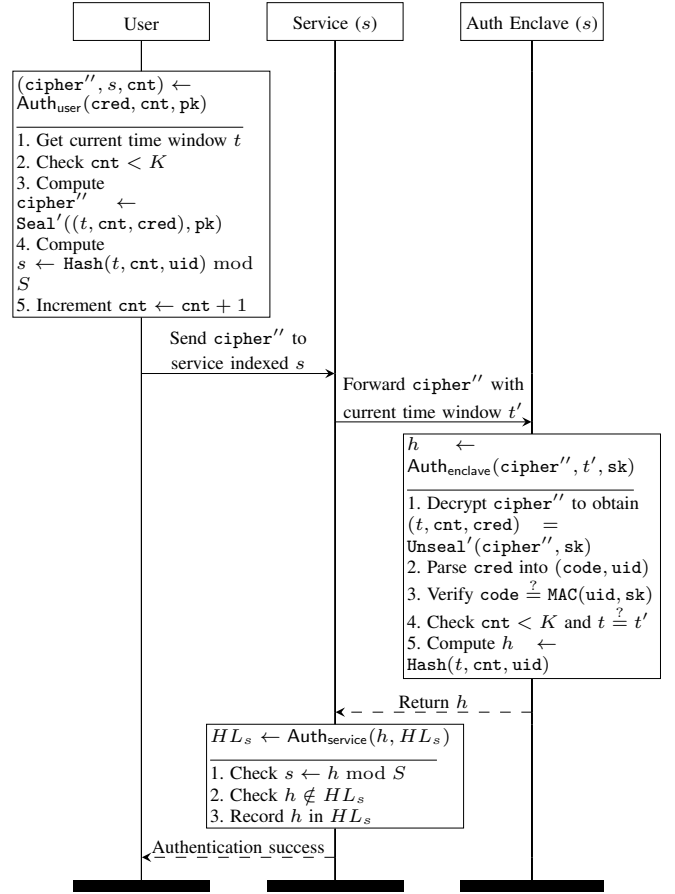

    \centering
    \begin{msc}[
        small values,
        draw frame=none,
        msc keyword=,
        head top distance=4mm,
        foot distance=0mm,
        instance distance=7.5mm,
        environment distance=8mm,
        level height=1.6mm,
        arrow scale=1.6,
        instance width=18mm,
        action width=34mm,
        action height=3mm,
        left inline overlap=11mm,
        right inline overlap=11mm,
        /tikz/font=\scriptsize,
        /tikz/line width=0.35pt,
        /msc/line width=0.35pt,
        label distance=0.3ex,
        every action/.style={align=left}
    ]{}

        \declinst{u}{}{User}
        \declinst{s}{}{Service ($s$)}
        \declinst{ae}{}{Auth Enclave ($s$)}

        \action{
            $(\cipher'', s, \cnt) \gets$\\[-0.3ex]
            $\auth_\user(\cred,\cnt,\pk)$\\[-0.3ex]
            \rule{30mm}{0.3pt}\\[-0.3ex]
            1.~Get current time window $t$\\
            2.~Check $\cnt < K$\\
            3.~Compute\\[-0.3ex]
            $\cipher'' \gets \enc'((t,\cnt,\cred),\pk)$\\
            4.~Compute\\[-0.3ex]
            $s \gets \hash(t, \cnt, \uid) \bmod S$\\
            5.~Increment $\cnt \leftarrow \cnt + 1$
        }{u}

        \nextlevel[25]

        \mess{\shortstack{Send $\cipher''$ to\\service indexed $s$}}{u}{s}

        \nextlevel[4]

        \mess{\shortstack{Forward $\cipher''$ with\\current time window $t'$}}{s}{ae}

        \nextlevel[1]

        \action{
            $h \gets \auth_\enclave(\cipher'', t', \sk)$\\[-0.3ex]
            \rule{30mm}{0.3pt}\\[-0.3ex]
            1.~Decrypt $\cipher''$ to obtain\\
            $(t,\cnt,\cred)=\dec'(\cipher'',\sk)$\\
            2.~Parse $\cred$ into $(\code,\uid)$\\
            3.~Verify $\code \isequal \mac(\uid,\sk)$\\
            4.~Check $\cnt < K$ and $t \isequal t'$\\
            5.~Compute $h \gets \hash(t, \cnt, \uid)$
        }{ae}

        \nextlevel[23]

        \mess*{Return $h$}{ae}{s}

        \nextlevel[1]

        \action{
            $HL_s \gets \auth_\service(h,HL_s)$\\[-0.3ex]
            \rule{30mm}{0.3pt}\\[-0.3ex]
            1.~Check $s \gets h \bmod S$\\
            2.~Check $h \notin HL_s$\\
            3.~Record $h$ in $HL_s$
        }{s}

        \nextlevel[11]

        \mess*{Authentication success}{s}{u}

    \end{msc}

    \caption{Authentication protocol.}
    \label{fig:authentication}
\end{figure}

%% file: text/7_implementation.tex
\section{Implementation}
To demonstrate its feasibility, we built an end-to-end proof-of-concept prototype of \sys{} in Rust, consisting of the Core Library, Rate-Limiting Service, protected Web App Server, and Client-side Proxy (see Fig.~\ref{fig:impl} for the overview). 

During the implementation, we made several assumptions.
We set the time-window length at one minute and $K$ to $100$ as the default setting.
Additionally, we abstracted the ID-verification mechanism since it can be realized through various means (e.g., TPM-backed certificates).

\begin{figure}
    \centering
    \footnotesize
    \includegraphics[width=0.8\linewidth]{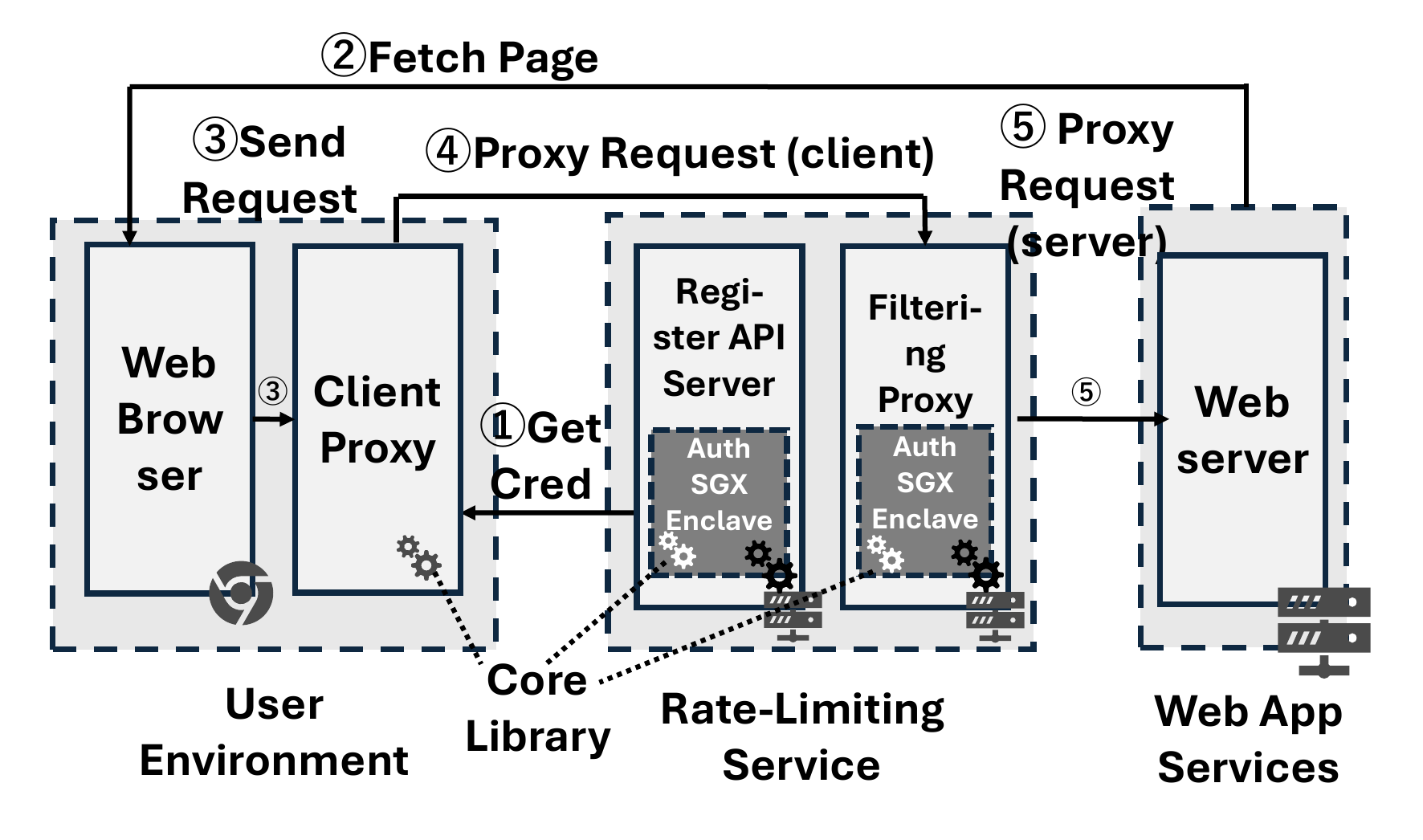}
    \caption{Implementation of \sys}
    \label{fig:impl}
\end{figure}

\paragraph*{Challenge and Solution: How to Deliver Ciphertexts}
Although \sys{} requires users to transmit $\cipher''$, determining the optimal transmission method is not trivial.
The most straightforward way is to encapsulate it in the HTTP header.
However, this allows the server to obtain $\cipher''$ only after the TLS handshake, so the risk of TLS flooding remains.

To address this issue, we use a TLS extension to deliver $\cipher''$.
This way, the service obtains $\cipher''$ when it receives the TLS ClientHello message, allowing it to filter user requests even before completing the TLS handshake.

However, the implementation choice described above introduces an additional risk; intermediary entities between users and the service, e.g., ISPs, may intercept $\cipher''$, allowing them to block user requests and/or send them as their own requests.
Although this is a limitation of our implementation, the impact of such attacks is limited.
    In the case of malicious blocking, users can detect the event (e.g., by not receiving a valid response), and resend $\cipher''$ via a different network route using the next counter value.
    Although this consumes an additional $\cnt$ value, the impact is minimal if $K$ is sufficiently large or the epoch length is short.
    Moreover, sending an intercepted request does not affect the rate-limiting functionality of \sys{}, as it does not increase the total number of requests that can be submitted to a service. This is because the intercepted request must be generated by a user who received a legitimate credential from the service, which controls the number of credentials it issues depending on the load it can handle.

\subsection{Core Library}
The Core Library is responsible for the main cryptographic operations.
This library covers most of the processes of the Auth Enclave and the user shown in Section~\ref{sec:protocol}, except for I/O-related processes (e.g., checking storage and sending data over the Internet).

Implementing this component required several technical design choices.
    First, we had to select cryptographic primitives that met the requirements set in Section~\ref{sec:system}.
    Specifically, we adopt SHA-256 for the hash function and HMAC-SHA256 for the MAC.
    For communication between the Auth Enclave and the user, we use a hybrid design in which HPKE is used to deliver an AEAD key (together with other messages), after which an AEAD is used.
    As for the primitives underlying these constructions, we adopt X25519 and ChaCha20-Poly1305, which are well suited to implementations with secret-independent control flow and memory access.
    We use implementations designed to avoid secret-dependent branches and memory accesses~\footnote{https://github.com/symbolicsoft/hpke-ng}~\footnote{https://docs.rs/chacha20poly1305/0.10.1/chacha20poly1305/}~\footnote{https://docs.rs/hmac/0.12.1/}.

Additionally, we implement the Core Library so that control flow and memory-access patterns do not depend on secret information.
    Aside from external cryptographic libraries, which we assume do not introduce secret-dependent branches or memory accesses, the only operation that required explicit handling was byte-string comparison.
    Therefore, we implemented all byte-string comparisons using a linear scan over all bytes.

\subsection{Rate-Limiting Service}
The Rate-Limiting Service is responsible for blocking excessive requests from reaching the Web App Server (described below), which consists of (i) a Registration API Server that provides the registration procedure, (ii) a Filtering Proxy that performs authentication, and (iii) an Auth SGX Enclave, which is an implementation of Auth Enclave.

\taggedpara{Registration API Server}
This server handles user registration. It relays the user's registration request to the Auth Enclave and returns the encrypted credential produced by the enclave to the user.
For compatibility with Pingora, we adopted Pingora Web for this implementation.

\taggedpara{Filtering Proxy}
This proxy is responsible for authentication.
Upon receiving a user request, the Filtering Proxy invokes the Auth Enclave and checks whether the blind token returned by the enclave has not been used before using a HashMap.
Taking inspiration from WAAP practices, we leverage a proxy design for the \sys{} authentication layer.
For the actual implementation, we adopted Pingora~\footnote{https://blog.cloudflare.com/pingora-open-source/} due to its scalability and Cloudflare's deployment track-record. The duplicate check and insertion into the HashMap are performed atomically under a mutex to prevent race conditions between concurrent requests.

One challenge in implementing the core proxy is extending TLS.
To address this challenge, we use OpenSSL to take advantage of its useful hook for validating custom extensions, namely \texttt{TlsSettings::set\_client\_hello\_callback}~\footnote{\label{fn:ssl-acceptor-builder}https://docs.rs/openssl/latest/openssl/ssl/struct.SslAcceptorBuilder.html}~\footnote{https://docs.rs/pingora/latest/pingora/listeners/tls/struct.TlsSettings.html}.

\taggedpara{Auth SGX Enclave}
\label{sec:impl:key-share}
This is the implementation of the \sys{} Auth Enclave.
    This provides the Filtering Proxy with an interface for invoking Auth Enclave functions implemented in the Core Library.
    Additionally, it provides the features of attesting $\pk$ and sharing $\sk$ between enclaves after mutual attestation~\footnote{This is done via a dedicated key-managing enclave.}, which is necessary for realizing the $HL$ sharding design discussed in Section~\ref{sec:sharding}.
As the execution environment for the Auth Enclave, we adopted Intel SGX because it provides RA and ME while keeping the TCB small.
    Furthermore, to preserve the Rust implementation while minimizing the TCB, we used Rust SGX and implemented the enclave as an externally callable library.
    This design places many operations, including communication handling, outside the enclave, further reducing the TCB.

\subsection{Web App Server}
The Web App Server returns HTML pages to users, and is protected by the Rate-Limiting Service.

\subsection{Client-side Proxy}
The Client-side Proxy enables the user's browser to access the Web App Server.
When a user sends a request through this proxy, the proxy increments the counter $\cnt$, then generates and attaches $\cipher''$ to the TLS extensions. For this implementation, we use OpenSSL, its hook function for TLS extension \texttt{SslConnectorBuilder::add\_custom\_ext}\footnote{https://docs.rs/openssl/latest/openssl/ssl/struct.SslConnectorBuilder.html}, and Pingora to align with the Filtering Proxy implementation.

%% file: text/8_evaluation.tex
\section{Performance Evaluation}

In this section, we evaluate the performance of \sys{}.
All experiments were conducted on a Standard DC24s v3 Azure VM, which has 24 CPU cores running at a frequency of 2.8 GHz, and 196 GB of memory.

\subsection{Latency/Throughput Evaluation}
\subsubsection{Methodology}
We measure the computational latency and throughput of \sys{} as well as two prior systems and compare the results.
    As for the prior works, we chose ARC~\cite{ARC}~\footnote{https://github.com/cloudflareresearch/arc} (under the $K=100$ setting) and Scrappy~\cite{Scrappy}, as they are recent anonymous rate-limiting systems whose code is open source.\footnote{Scrappy uses akakou/ecdaa; we used its Rust implemented version.}
    We also chose these systems because they are secure against colluding authorities (cf. Section~\ref{sec:related-work})~\footnote{Our measurements for each scheme exclude the cost of duplication check from each total latency, as they do not provide the same implementation of $HL_s$ as \sys{}.
    This setting favors the prior systems and is unfavorable to \sys{}.}.

\subsubsection{Results}
\taggedpara{\sys{} latency and throughput}
Tables~\ref{tab:baseline-latency-client}, \ref{tab:baseline-latency-server}, and \ref{tab:prior-compare} show our measured results, which confirm that \sys{} operates at high performance.
User-side latency is $\clientSumLatency$ while service-side latency is $\serverSumLatency$ (+ $\dedupeCheckLatency$ for duplicate checking and inserting).
    We refer the reader to Tables~\ref{tab:baseline-latency-client} and \ref{tab:baseline-latency-server} for a breakdown of the client- and server-side latencies, respectively.
Furthermore, \sys{} achieves $\totalThroughput{}$ (connections per second) for verification, demonstrating its scalability.

\taggedpara{Comparison evaluation}
We demonstrate that \sys{} is more efficient than prior work.
More concretely, under the above conditions, the proving latencies of ARC and Scrappy are $\arcSignLatency$ and $\scrappySignLatency$, respectively.
Verification latency is $\arcVerifyLatency$ and $\scrappyVerifyLatency$ for ARC and Scrappy, respectively, which leads to a throughput of $\arcThroughput$ and $\scrappyThroughput$ for ARC and Scrappy, respectively.
Based on the measurement numbers shown above, the proving latency of \sys{} is approximately \clientFasterThanExisting$\times$ or more faster than the two related works and \serverFasterThanExisting$\times$ or more efficient for verification.

\begin{table}[t]
    \centering
    \caption{Client-side \sys{} latency.}
    \begin{tabular}{|l|c|} \hline
        Process  & Latency  \\ \hline \hline
        Encryption & {\encryptLatency} \\ \hline
        Route selection & {\routeChooseLatency}\\ \hline \hline
        Total & \clientSumLatency \\ \hline
    \end{tabular}
    \label{tab:baseline-latency-client}
\end{table}

\begin{table}[t]
    \centering
    \caption{Server-side \sys{} latency. Parenthesized text indicates duplicate-check and inserting process latencies.}
    \begin{tabular}{|l|c|} \hline
        Process  & Latency  \\ \hline \hline
        Decryption & {\decryptLatency}\\ \hline
        Verification & {\verificationLatency}\\ \hline
        Token generation & {\tokenGenerateLatency}\\ \hline
        Route selection & {\routeChooseLatency}\\ \hline
        Call Enclave \& Return & \callLatency \\ \hline 
        Others & \otherLatency \\ \hline \hline
        Total & \serverSumLatency{} ( +\dedupeCheckLatency) \\ \hline 
    \end{tabular}
    \label{tab:baseline-latency-server}
\end{table}

\begin{table}[t]
    \centering
    \newcolumntype{P}[1]{>{\centering\arraybackslash}p{#1}}
    \caption{Latency and throughput comparison with prior works. cps = connections per second.
    }
    \begin{tabular}{|c|P{4em}|P{5em}|c|}  \hline
        & Proving Latency & Verification Latency & Throughput \\ \hline \hline
        ARC & \arcSignLatency & \arcVerifyLatency & \arcThroughput \\ \hline
        Scrappy & \scrappySignLatency & \scrappyVerifyLatency & \scrappyThroughput \\ \hline
        \sys{} & \clientSumLatency & \serverSumLatency & \totalThroughput{}  \\ \hline
    \end{tabular}
    \label{tab:prior-compare}
\end{table}

\subsection{Filtering Performance Evaluation}
\subsubsection{Methodology}
We evaluate the latency and throughput of the filtering process of the filtering proxy.
    Recall that in \sys{}, this proxy is responsible for filtering out any invalid requests \emph{before} performing the TLS handshake to mitigate TLS flooding attacks.
Therefore, in this evaluation, we focus only on the computational cost of receiving TLS ClientHello messages (i.e., \texttt{SslAcceptor::accept}~\footnote{https://docs.rs/openssl/latest/openssl/ssl/struct.SslAcceptor.html}) with a custom extension configured to use \texttt{TlsSettings::set\_client\_hello\_callback}~\footnote{See footnote~\ref{fn:ssl-acceptor-builder}.}. %
    In detail, we investigate the following three potential attack scenarios: (i) a missing \sys{} TLS extension header; (ii) using $\cipher''$ generated from an invalid credential; and (iii) replaying  $\cipher''$.
We also measure the cost of establishing a TLS connection and compare that with the above measurements.

\subsubsection{Results}
Our experiments show that \sys{} provides sufficient throughput to mitigate the attack described above.
Table~\ref{tab:base-compare} summarizes the results.
The filtering proxy achieves (i) $\tlsMissExtensionThroughput{}$, (ii) $\tlsInvalidCredThroughput{}$, and (iii) $\tlsReplayCredThroughput{}$ for the three scenarios.
This makes \sys{} \throughputLargerThanTLS{} times more scalable than the measured TLS handshake, which was \tlsThroughput{}.
Additionally, Figures~\ref{fig:each-attack} and \ref{fig:vs-tls} show that the throughput of \sys{} increases linearly.
    We can also observe the same trend in an environment with multiple \sys{} servers.

\begin{table}[t]
    \centering
    \caption{Latency and throughput comparison with TLS under a single thread setting. cps = connections per second.
    }
    \begin{tabular}{|c|c|c|c|c|}  \hline
         & TLS  & (i) & (ii) & (iii) \\ \hline 
        Latency  &  \tlsLatency & \tlsMissExtensionLatency & \tlsInvalidCredLatency & \tlsReplayCredLatency \\ 
        \hline
        Throughput  & \tlsThroughput & \tlsMissExtensionThroughput{} & \tlsInvalidCredThroughput{} & \tlsReplayCredThroughput{}  \\ \hline
    \end{tabular}
    \label{tab:base-compare}
\end{table}

\begin{figure}[t]
    \centering
    \includegraphics[width=0.7\linewidth]{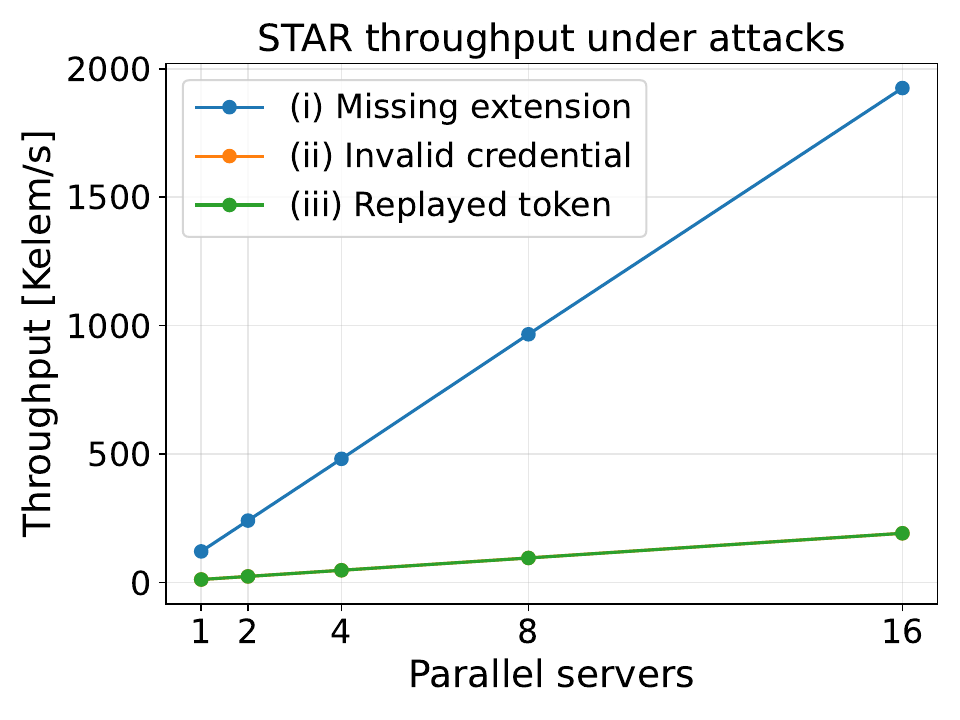}
    \caption{Throughput comparison for each attack. The line of (iii) overlays (ii).}
    \label{fig:each-attack}
\end{figure}

\begin{figure}[t]
    \centering
    \includegraphics[width=0.7\linewidth]{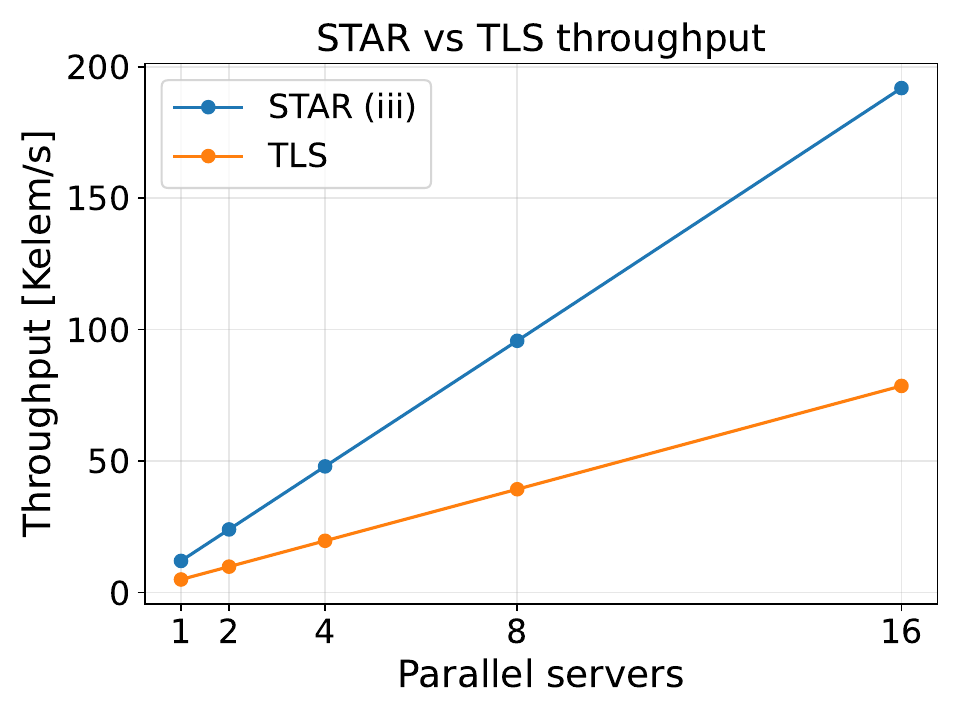}
    \caption{Throughput comparison with responding to TLS ClientHello messages}
   
    \label{fig:vs-tls}
\end{figure}

\subsection{Other Costs}
\subsubsection{Bandwidth Consumption}
We measured the bandwidth consumption of the authentication protocol.
The one-way overhead of \sys{} during authentication is \payloadSize{}.
In comparison, the bandwidth consumption of ARC and Scrappy is $1248 (= 352 + 128 \times \lceil\log_2 (100)\rceil)$ \si{\byte} and \SI{261}{\byte}, respectively, making the communication overhead of \sys{} lower than that of previous works.

\subsubsection{Memory Consumption}
We now examine the memory consumption of \sys{}.
From Section~\ref{sec:protocol} we can see that the storage size of $HL_s$ increases linearly.
To empirically evaluate the memory consumption of $HL_s$ in a modeled real-world scenario, we consider a Wikipedia-scale deployment.
Concretely, we assume approximately 2,400 servers\footnote{{https://wikitech.wikimedia.org/wiki/Servers}} and 22.6 billion accesses per month\footnote{https://wikimedia.org/api/rest\_v1/metrics/pageviews/aggregate/all-projects/all-access/all-agents/monthly/2026040100/2026043000}, with the assumption that 1\% of the servers are allocated for \sys{-related} tasks.
In this setting, we estimate the number of entries in $HL_s$ for a single time window to be $\HLSize$ $\left(=22.6~\text{billion} \div (0.01 \times 2400 \times 30 \times 24 \times 60)\right)$, which totals approximately \HLStorage{}.
This is practical and because $HL_s$ is flushed in a batch process when $t$ turns over, the memory overhead is constant throughout the operation and therefore remains within a practical operating range.

%% file: text/9_security.tex
\section{Security Analysis}
\label{sec:security}
In this section, we analyze the security of the \sys{} protocol.
We first conduct an access-pattern leakage analysis for any secret-dependent control flow or memory accesses and then, based on its results, provide proof sketches and full proof in Appendix~\ref{sec:full-security-proof} showing that the \sys{} protocol satisfies the security requirements defined in Section~\ref{sec:system}.

\subsection{Access-Pattern Leakage Analysis}

\subsubsection{Experimental Setup}
We use a dynamic analysis tool to validate that our Core Library (and its dependencies) does not introduce secret-dependent branches or secret-dependent memory accesses. 
    Note that in this evaluation, we evaluate the Core Library functionality related to the \sys{} protocol \emph{outside} the SGX environment, where the required analysis tools are available. 
    Therefore, SGX-specific functionalities that require an SGX environment, including remote attestation and sealing, remain outside the scope of this analysis.

This work adopts a Timecop-style~\cite{Timecop} methodology. %
This approach repurposes Valgrind's mechanism for identifying conditional jumps and memory accesses that depend on uninitialized values to detect those that depend on secret values.
Valgrind performs this analysis at byte-level granularity, which is finer than cache-line or page-level access observations and therefore also covers these coarser access granularities.
This work uses crabgrind, a Valgrind wrapper.

In our experiment, we treat all inputs for the enclave (e.g., $\sk$, $t$, $\cipher$, $\cipher''$) as the traceable variables and execute setting up, credential issuance, credential verification, and token generation as our target functions.
Although we only set the above values, Valgrind automatically traces all values generated from them, i.e., $\uid$, $\cnt$ and $\cred$. We then inspect the conditional branches and access patterns within the Core Library that depend on these secret values. 

Note that these checks can raise false alarms, i.e., secret-dependent branches that do not affect security.
For example, even if a user supplies an invalid ciphertext, credential, or ephemeral key to an error-handling branch, the branch leaks that the input was invalid, which is acceptable in our setting.
We therefore manually inspect each reported conditional branch and its corresponding code, and confirm that the branch does not indeed leak sensitive information beyond what the adversary already legitimately learns.

\subsubsection{Experimental Results}
\label{sec:dynamic-analysis-results}
We did not identify security-relevant secret-dependent memory accesses or branches.
Although Valgrind reported several secret-dependent conditional branches or access-pattern changes, our manual inspection of the relevant functions confirmed that all reported cases were false positives.
The list of functions and their accompanying inspection result is shown in Appendix~\ref{sec:valgrind-results}.

\subsection{Cryptographic Security Analysis}
\label{sec:crypto-analysis-sketch}

\begin{theorem}[Rate-Limitability of the \sys{} protocol]
Under the random oracle model, assuming IND-CCA security of AEAD and public-key encryption, and EUF-CMA security of the MAC scheme, \sys{} satisfies rate-limitability. 
\end{theorem}

\begin{proof}[Proof Sketch] We prove the rate-limitability of \sys{} with the following sketched hybrid game:  \\
\taggedpara{Game 1. Ciphertext-Level Manipulation} We first consider arbitrarily crafted ciphertexts. Malformed $\cipher$ or $\cipher''$ is rejected during registration or authentication. Otherwise, $\cipher$ and $\cipher''$ decrypt to some $\ck^*$ and $(t^*, \cnt^*, \cred^*)$, respectively, and the security of the underlying encryption schemes ensures that ciphertext manipulation gives no advantage beyond choosing or reusing these plaintext values.\\
\taggedpara{Game 2. Plaintext-Level Manipulation} We next consider adversarially chosen plaintexts. Since $\pk$ is public, the adversary may encrypt an arbitrary $\ck^*$, but this only creates an additional registration counted in $N$. A forged credential, an invalid $\cnt$ or $t$, or an incorrect shard $s$ is rejected by MAC verification, the range and time-window checks, or the check $s = h \bmod S$, respectively.\\
\taggedpara{Game 3. Replay Attacks} Finally, replaying the same $(t, \cnt, \cred)$ produces the same $h$ and is rejected because $h \in HL_s$. Hence, for each credential, at most $K$ distinct counter values can be accepted in any time window. Therefore, after $N$ registrations, the adversary can successfully authenticate at most $K \times N$ times in each time window.
\end{proof}

\begin{theorem}[Unlinkability of the \sys{} protocol]
Under the random oracle model, assuming IND-CCA security of AEAD and public-key encryption, INT-CTXT security of AEAD, EUF-CMA security of the MAC scheme, and the fact that $\LeakAccessPattern(\register_\enclave, \cdots), \LeakAccessPattern(\auth_\enclave, \cdots)$ only leaks negligible information, \sys{} satisfies unlinkability.
\end{theorem}
\begin{proof}[Proof Sketch]
We prove the unlinkability of \sys{} via the following sequence of hybrid games: \\
\taggedpara{Game 1. Invalid Ciphertext}
We first consider malformed ciphertexts submitted to the oracles.
By the integrity and security of the underlying encryption schemes, such ciphertexts are rejected except with negligible probability and therefore reveal no information about the challenge bit. \\
\taggedpara{Game 2. Invalid Plaintext}
Since $\pk$ is public, the adversary can encrypt arbitrary plaintexts, including arbitrary $(t,\cnt,\cred)$ values.
However, a fresh $\cred$ is accepted only if it contains a valid MAC, whose forgery succeeds with only negligible probability. \\
\taggedpara{Game 3. Invalid $t'$}
If the adversary submits an arbitrary time window $t'^*$, the request is processed normally if $t = t'^*$; otherwise, $\bot$ is returned. This verification reveals only whether the submitted time window matches $t$, which is common to all legitimate users in our model. Since this information is independent of the challenge bit, it provides no distinguishing advantage. \\
\taggedpara{Game 4. Replacing Random-Oracle Outputs}
We replace each $h=\hash(t,\cnt,\uid)$ with a consistently lazy-sampled random value.
Because $\uid$ is hidden and has sufficiently high entropy, the adversary queries the corresponding random-oracle input only with negligible probability.
Moreover, distinct authentications use distinct $(t,\cnt,\uid)$ tuples, so the resulting $h$ values, and hence $s=h\bmod S$, provide no advantage in linking requests. \\
\taggedpara{Game 5. Replacing Ciphertexts}
We replace each challenge authentication ciphertext with an encryption of a fixed-length dummy plaintext, while answering repeated challenge-ciphertext queries consistently using the $h$ values sampled in Game~4.
By the security of the public-key encryption scheme, the adversary cannot distinguish the real and dummy ciphertexts except with negligible probability. \\
\taggedpara{Game 6. Remaining Values}
The remaining non-random values are the public user-state handle $n$ and the success status, neither of which depends on the challenge bit.
Furthermore, by the assumed security of $\operatorname{Leak}_{\mathsf{access}}$, the observed access patterns reveal only negligible information.
Thus, the challenge distributions are identical in the final game, and the adversary's success probability in the original game is at most $1/2+\negl$.

\end{proof}

%% file: text/10_discussion.tex
\section{Discussion}

\label{sec:discuss}

\subsection{TEE Timing \& Cache-line Side-channel Attacks}
\sys{} is designed to mitigate any access-pattern leakage; however, a TEE-based system may also leak other side channel information, including timing and cache-line side channels. 
Although these attacks are treated outside the scope of this work and \sys{} does not explicitly guarantee resistance to such attacks, its design may reduce the available attack surface.
This is because, although we do not conduct any detailed security analysis against these attacks, \sys{} does not include potentially dangerous secret-dependent control flow and memory accesses, which many timing and cache-line side-channel attacks rely on.

\subsection{Protocol-level Timing-Correlation Attack}
\sys{} provides stronger mitigation than Privacy Pass against attacks that correlate credential replenishment with subsequent use. 
    Privacy Pass users must periodically replenish their passes, potentially while revealing their identities. Although the passes themselves cannot be directly linked to this identity information, a colluding service and pass issuer may exploit timing correlations between replenishment and spending events to link multiple spending events to one another as originating from the same user. 
    This risk is particularly pronounced when a pass is spent shortly after replenishment, because the temporal proximity provides a stronger signal for associating subsequent spending events with the same user.
In contrast, \sys{} has no pass-replenishment phase and therefore exposes no corresponding pair of events to correlate.
This difference gives \sys{} a concrete advantage over Privacy Pass under collusion between the service and authorities (see Sections~\ref{sec:related-work} and \ref{sec:system}).\footnote{
    In \sys{}, registration may be linked to a user's first authentication request, but this does not enable linking multiple authentication events. Similar to other anonymous rate-limiting and authentication protocols, including Scrappy~\cite{Scrappy}, we consider this correlation out of scope.
}

\subsection{Tracking Users via Clock Skew}
\sys{} assumes that the user clocks are synchronized. This is because users whose clocks exhibit drift may be tracked by the server if that drift is exploited.
An alternative design is to use a service-provided time period $t$, an approach leveraged by related works such as Scrappy~\cite{Scrappy} and CACTI~\cite{CACTI}.
However, introducing this approach into \sys{} would approximately double \sys{'s} communication costs, resulting in a trade-off between privacy and latency.

\section{Limitation}
A limitation of \sys{} is that its security depends on the security of the underlying TEE. Additionally, our access-pattern leakage analysis covers secret-dependent branches and memory accesses in the protocol logic and excludes other leakage patterns.

%% file: text/11_conlusion.tex
\section{Conclusion}

In this work, we proposed \sys{}, a fast and anonymous rate limiter that utilizes TEEs on the server side. The \sys{} protocol avoids secret-dependent control flow and memory-access patterns that could reveal user-related secrets and its servers are distributed with the sharding mechanism.

%% file: text/C_notation.tex
\subsection{Notation}
\label{sec:notation}
The notation used throughout this work is summarized in Table~\ref{tab:notation}.

%% file: text/D_proceedure.tex
\subsection{\sys{} Procedure Details}
\subsubsection{Registration}

\begin{enumerate}
    \item The user executes the first user-side registration function as follows: $\pmb{(\cipher, \ck) \gets \register_{\user,1}(\pk)}$. The function randomly generates an ephemeral symmetric key $\ck$ and encrypts it under the Auth Enclave's public key to produce $\cipher \leftarrow \enc'(\ck, \pk)$, where $\enc'$ is asymmetric encryption.
    \item The user transmits their identity information $\identity$ along with the ciphertext $\cipher$ to the service.
    \item The service atomically verifies that $\identity \notin IL$ and stores $\identity$ in $IL$.
    \item The service invokes the Auth Enclave, which executes the enclave-side registration function as follows: $\pmb{\cipher' \gets \register_\enclave(\cipher, \sk)}$. The function generates the encrypted user credential $\cipher'$ as follows:
    \begin{enumerate}
        \item The Auth Enclave randomly generates a unique user ID $\uid$.
        \item The Auth Enclave computes the authentication code $\code \leftarrow \mac(\uid, \sk)$.
        \item The Auth Enclave binds the credential as $\cred \leftarrow (\code, \uid)$.
        \item The Auth Enclave decrypts $\cipher$ to recover the symmetric key $\ck \leftarrow \dec'(\cipher, \sk)$.
        \item The Auth Enclave encrypts the provisioned data to produce $\cipher' \leftarrow \enc(\cred, \ck)$, where $\enc$ is authenticated symmetric key encryption (i.e., AEAD).
    \end{enumerate}
    \item The Auth Enclave returns $\cipher'$ to the user via the service.
    \item The user executes the second user-side registration function as follows: $\pmb{\cred \gets \register_{\user,2}(\cipher', \ck)}$. The function locally decrypts $\cipher'$ using $\ck$ to retrieve $\cred$, i.e., $\cred \leftarrow \dec(\cipher', \ck)$.
\end{enumerate}

\subsubsection{Authentication}

\begin{enumerate}
    \item The user executes the user-side authentication function as follows: $\pmb{(\cipher'', s, \cnt) \gets \auth_\user(\cred, \cnt, \pk)}$. The function constructs the authentication payload $\cipher''$ as follows:
    \begin{enumerate}
        \item The user queries their local clock to obtain the current time window $t$.
        \item The user verifies that the counter respects the threshold limit: $\cnt < K$.
        \item The user encrypts the authentication data using the Auth Enclave's public key: $\cipher'' \leftarrow \enc'((t, \cnt, \cred), \pk)$.
        \item The user computes the service index using the current time window, counter, and user id $s \gets \hash(t, \cnt, \uid) \bmod S$.
        \item The user increments their local session counter: $\cnt \leftarrow \cnt + 1$.
    \end{enumerate}
    \item The user transmits $\cipher''$ to the service indexed by $s$.
    \item The service forwards $\cipher''$ and current time $t'$ to the Auth Enclave for verification.
    \item The Auth Enclave executes the enclave-side authentication function as follows: $\pmb{h \gets \auth_\enclave(\cipher'', t', \sk)}$. The function decrypts and validates the authentication payload as follows, using the service-provided $t'$ for the time-window check:
    \begin{enumerate}
        \item The Auth Enclave decrypts the payload to extract $t, \cnt$, and $\cred$: $(t, \cnt, \cred) \leftarrow \dec'(\cipher'', \sk)$.
        \item The Auth Enclave parses $\cred$ into $(\code, \uid)$ and verifies the cryptographic integrity of the credential by checking $\code \isequal \mac(\uid, \sk)$.
        \item The Auth Enclave asserts that the counter satisfies $\cnt < K$.
        \item The Auth Enclave checks equality of time-windows: $t \isequal t'$.
        \item Upon successful verification, the Auth Enclave generates a blind token $h \leftarrow \hash(t, \cnt, \uid)$ and returns $h$.
    \end{enumerate}
    \item The service executes the atomic service-side authentication function as follows:
    $\pmb{HL_s \gets \auth_\service(h,HL_s)}$.
    The function performs the following steps:
    \begin{enumerate}
        \item The service checks that the request is routed to the correct service index
        $s = h \bmod S$.
        \item The service checks that the token has not been previously consumed
        ($h \notin HL_s$); otherwise, it returns $\bot$.
        \item The service records $h$ in the duplicate-check list $HL_s$.
\end{enumerate}
    \item Finally, the service returns authentication success to the user.
\end{enumerate}

%% file: table/notation.tex
\begin{table}[t]
    \footnotesize
    \caption{Notation used in this work and its meaning}
    \label{tab:notation}
    \centering
    \begin{tabular}{ll}
        \hline
        \textbf{Symbol} & \textbf{Meaning / Description} \\
        \hline
        $\sk$ & Secret key held exclusively by the Auth Enclave \\
        $\pk$ & Public key corresponding to $\sk$ \\
        $\ck$ & Ephemeral symmetric key \\
        $\mac(\cdot,\cdot)$ & Message authentication code generation \\
        $\hash(\cdots)$ & Hash function \\
        $\enc(\cdot,\cdot)$ & Authenticated symmetric key encryption (i.e., AEAD) \\
        $\dec(\cdot,\cdot)$ & Authenticated symmetric decryption \\
        $\enc'(\cdot,\cdot)$ & Public-key encryption \\
        $\dec'(\cdot,\cdot)$ & Public-key decryption \\

        $\cnt$            & Counter \\
        $\code$           & MAC value \\
        $\cred$           & User credential \\
        $\cipher$         & Ciphertext \\
        $h$               & Blind token \\
        $t$               & Current time window \\
        $\uid$            & User ID \\
        $\identity$       & User identity information (e.g., hardware id) \\
        $HL_s$            & Set of previously seen $h$ values by server $s$ \\
        $IL$              & Set of previously seen $\identity$ values \\
        $K$               & Upper bound of the counter \\
        $s$               & Index of server \\
        $S$               & Number of servers \\
        $N$               & Number of registered users \\
        $n$               & User index \\
        \hline
    \end{tabular}
\end{table}

%% file: text/B_side-channel-vector.tex
\subsection{Secret-dependent conditional branches and access-pattern changes}
\label{sec:valgrind-results}

\begin{enumerate}
\label{sec:secret-depend-branch}
    \item
    \verb| <hpke_ng::kem::dh::X25519 as | \\
    \verb|hpke_ng::kem::dh::DiffieHellman>::dh| 

    This function checks whether the Diffie-Hellman output is an all-zero value. Even if the service impersonates a user and triggers this condition by using an invalid public key as the user-side ephemeral public key, this would not help the provider recover any other sensitive information. Thus, it does not constitute a meaningful leak.
    \label{enum:secret-branch:decap}

    \item
    \verb|<hpke_ng::aead::ChaCha20Poly1305 as | \\
    \verb|hpke_ng::aead::SealingAead>::open|

    This function branches on the result of the MAC check. Revealing whether MAC verification succeeds does not leak any sensitive information.
    \label{enum:secret-branch:decrypt_in_place_detached}

    \item
    \verb|star_core::serv::enclave::Enclave|\\
    \verb|::issue_token_plain|

    This function has a conditional branch that checks whether $\cnt < K$, whether $t$ matches the current time window, and whether $\cred$ is valid. The results of these checks only leak information that is safe to reveal.
\end{enumerate}

%% file: text/A_proof.tex
\subsection{Full Security Proof}
\label{sec:full-security-proof}
\subsubsection{Rate-Limitability of \sys}

\begin{proof}
In this proof, we show that the adversary cannot win the rate-limitability game with non-negligible probability. We prove this via a hybrid game showing that an authentication requires a valid credential and that the valid credential enables a limited number of authentications.

For clarity, let $T$ denote the responses returned by the oracles to the adversary, and let $R$ denote the corresponding queries submitted by the adversary. The oracle responses visible to the adversary are 
\begin{align*}
    T_\ratelimit^\register = \cipher', 
    T_\ratelimit^\auth = 1 \\ 
\end{align*}
for adversarially chosen oracle queries of the form
\begin{align*}
    R_\ratelimit^\register = \cipher, 
    R_\ratelimit^\auth = (\cipher'', s). \\ 
\end{align*}
We show that none of these oracle interactions gives the adversary a non-negligible advantage in generating valid $R_\ratelimit^\auth$ $i (> K \times N)$ times through the following sequence of hybrid games.

\begin{itemize}
\item \textbf{Game 1. Ciphertext-Level Manipulation:}
    We first consider arbitrary ciphertexts constructed by the adversary, including those adaptively derived from previous oracle interactions. In particular, the adversary may construct $\cipher$ using previously obtained $\cipher'$, or construct $\cipher''$ using previously obtained values such as $\ck$, $\cred$, or oracle responses.
    
    A malformed $\cipher$ or $\cipher''$ is rejected if public-key decryption fails and, therefore, cannot increase $N$ or $i$. Otherwise, $\cipher$ decrypts to some $\ck^*$, and $\cipher''$ decrypts to some tuple $(t^*,\cnt^*,\cred^*)$. Similarly, any value legitimately recovered from $\cipher'=\enc(\cred,\ck)$ may be reused in subsequent queries. However, any additional useful information about $\ck$, $\cred$, or secret key material obtained from these ciphertexts with non-negligible probability would contradict the IND-CCA security of the underlying encryption schemes.
    
    Thus, except with negligible probability, adaptive manipulation of $\cipher$, $\cipher'$, and $\cipher''$ gives the adversary no advantage beyond choosing or reusing the plaintext values $\ck$, $\cred$, $t$, and $\cnt$. We analyze manipulations of these values, together with the routing value $s$, in the next game.
        
    \item \textbf{Game 2. Plaintext-Level Manipulation:}
    We now consider valid ciphertexts that contain plaintexts chosen by the adversary.
    Since $\pk$ is public, the adversary can generate and submit a valid $\cipher$ encrypting an arbitrary key $\ck^*$.
    This is permitted; however, because $\ck^*$ is chosen and therefore known by the adversary, doing so merely allows the adversary to complete an additional registration.
    Each such registration is counted by a call to $\oracle_\ratelimit^\register$ and increments $N$, and therefore does not help the adversary exceed the final bound of $K \times N$ successful authentications.

    Similarly, a valid $\cipher''$ may encrypt an arbitrary tuple $(t,\cnt,\cred)$.
    However, for a fresh credential $\cred=(\code,\uid)$ to be accepted, it must satisfy
    \[
    \code=\mac(\uid,\sk).
    \]
    By the EUF-CMA security of the MAC scheme, the probability that the adversary successfully generates such a fresh valid credential is negligible.

    The adversary may also attempt to manipulate the time window value $t$ or counter value $\cnt$, for example, by submitting a time window that is different from the current one or a counter value outside the permitted range.
    However, any value $\cnt \geq K$ is rejected by the range check performed by $\auth_\enclave$.
    Similarly, a value of $t$ that does not correspond to the current time window is rejected by the time-window check performed by $\auth_\enclave$.
    Thus, except with negligible probability, every accepted authentication request contains a valid registered credential, the current time window, and a counter satisfying $\cnt < K$.

    Finally, the adversary may forge $s$ and redirect a request to another service $s' \neq s$ to bypass the duplicate check in $HL_s$. In particular, even if the adversary has already authenticated with shard $s$ where $h \in HL_s$, it may attempt to re-authenticate with the same $h$ at another shard $s'$ for which $h \notin HL_{s'}$. This attack fails because $\auth_\service$ verifies that the request was routed to the shard that satisfies $s = h \bmod S$. As we show above, for each authentication request $\auth_\enclave$ verifies the credential, checks $\cnt < K$ and the time window $t$, and computes $h$ from the validated values $(t,\cnt,\uid)$. Therefore, except with negligible probability, any forged routing value $s'$ is rejected before the adversary can bypass duplicate detection.
    
    \item \textbf{Game 3. Replay of Previously Used Values:}
    The remaining strategy available to the adversary is to replay a previously accepted tuple $(\cnt,t,\cred)$.
    However, replaying the same tuple causes $\auth_\enclave$ to derive the same blind token
    \[
    h = \hash(t, \cnt, \uid).
    \]
    Since this value has already been accepted in the past, it is recorded in $HL$, and therefore any subsequent requests that produce the same $h$ are rejected by the check $h \notin HL$.
    Hence, the adversary cannot successfully authenticate more than once using the same tuple $(\cnt,t,\cred)$.

    Combining the above games, for any fixed time window $t$ and credential $\cred$, each successful authentication must use a distinct counter value satisfying $\cnt < K$.
    Hence, each credential can be used for at most $K$ successful authentications within a single time window.
    The number of credentials available to the adversary is at most $N$, corresponding to the number of calls to $O_\ratelimit^\register$.
    Therefore, the total number of successful authentications within a single time window is bounded by $K \times N$.
    Thus, the probability that the adversary wins the rate-limitability game by achieving more than $K \times N$ successful authentications is negligible.
    Thus, \sys{} satisfies rate-limitability.
\end{itemize}

\end{proof}

\subsubsection{Unlinkability of \sys}

\begin{proof}
We show that any PPT adversary has only a negligible advantage in guessing the challenge bit $b$, even after making arbitrary queries to the permitted oracles. In the unlinkability game, the adversary receives $(\sigma_1,\sigma_2)$, where both authentication ciphertexts are generated from the same user id $\uid$, or $(\sigma_1,\sigma_3)$, where the latter authentication ciphertext is generated from a different user state. We show that these two distributions are computationally indistinguishable.

For clarity, let $T$ denote the responses returned by the oracles to the adversary, and let $R$ denote the corresponding queries submitted by the adversary. The oracle responses visible to the adversary are
\begin{align*}
T_\unlink^{\register_\user, 1}
    &= (\cipher, n), \quad
       T_\unlink^{\register_\user, 2} = 1, \quad
       T_\unlink^{\auth_\user} = (\cipher'', s),
\\
T_\unlink^{\register_\enclave}
    &= (\cipher', \AccessPatternLeakage),
T_\unlink^{\auth_\enclave}
    = (h, \AccessPatternLeakage).
\end{align*}
for adversarially chosen oracle queries of the form
\begin{align*}
    R_\unlink^{\register_\user, 2} &= (\cipher', n),   
    &R_\unlink^{\register_\enclave} &= \cipher,  \\
    R_\unlink^{\auth_\user} &= n, 
    &R_\unlink^{\auth_\enclave} &= (\cipher'', t').
\end{align*}
We show that none of these oracle interactions gives the adversary a non-negligible advantage in distinguishing the two challenge distributions through the following sequence of hybrid games.

\begin{itemize}
    \item \textbf{Game 1. Invalid Ciphertext: }
    Consider malformed ciphertexts submitted to the oracles. A $\cipher'$ given to the second registration oracle is accepted only if it is a valid AEAD ciphertext under the hidden key $\ck$ associated with $n$. By INT-CTXT security, forging a fresh valid ciphertext is possible only with negligible probability; otherwise, the query is rejected without modifying the user state. Similarly, malformed public-key ciphertexts $\cipher$ or $\cipher''$ are rejected if decryption fails. Hence, malformed ciphertexts reveal no information about the challenge bit except with negligible probability.

    \item \textbf{Game 2. Invalid Plaintext: }
    Now consider valid ciphertexts containing adversarially chosen plaintexts. Since $\pk$ is public, the adversary can submit a valid $\cipher$ encrypting an arbitrary key $\ck^*$. However, the resulting $\cipher'$ is encrypted under $\ck^*$ and cannot be decrypted by the second registration oracle for an honest index $n$ unless $\ck^*=\ck$. Thus, this cannot modify an honest user's state or reveal information about the challenge.

    Similarly, a valid $\cipher''$ may encrypt any $(t,\cnt,\cred)$, but acceptance with a fresh $\cred=(\code,\uid)$ requires $\code=\mac(\uid,\sk)$. By EUF-CMA security of the MAC, producing such a fresh credential succeeds only with negligible probability.
    
    \item \textbf{Game 3. Invalid $t'$: }
    The adversary may submit an arbitrary time window $t'^*$ to $\oracle_\unlink^{\enclave_\auth}$ in an attempt to obtain additional information.
    We consider two cases.
    If $t = t'^*$, the oracle processes the query normally.
    Otherwise, if $t \neq t'^*$, the time-window check fails and the oracle returns $\bot$.
    Thus, this check reveals only whether the supplied time window matches the current one.
    Since all honest users share the same current time window $t$, this outcome is independent of the user's identity and the challenge bit.
    Therefore, it provides no advantage in distinguishing the challenge distributions.
    
    \item \textbf{Game 4. Replacing Random-Oracle Outputs: }
    In this game, we replace each $$h=\hash(t,\cnt, \uid)$$ with the output of a lazily sampled random function. Since $\uid$ remains hidden from the adversary and has sufficient entropy, the probability that the adversary queries the exact corresponding input is negligible. The simulation remains consistent for repeated $(t,\cnt,\uid)$ tuples, while distinct challenge authentications have distinct tuples because $\cnt$ or $\uid$ differs. The resulting $h$ values of these tuples are therefore independent random values of equal length, up to negligible probability. For the same reason, the $s \gets h \bmod S$ values from these $h$ values would appear randomly to the adversary, and thus provide no advantage in linking authentication requests.
    
    \item \textbf{Game 5. Replacing Ciphertexts: }
    In this game, we replace each challenge authentication ciphertext with an encryption of a dummy plaintext that has the same fixed length as the authentication tuple. For consistency, if an exact challenge ciphertext is submitted to $\oracle_\unlink^{\auth_\enclave}$, the oracle returns the corresponding randomly sampled \(h\) (or $\bot$) defined in Game 4; all other oracle queries are processed normally.
    
    By the IND-CCA security of the public-key encryption scheme, the adversary cannot distinguish encryptions of the real authentication tuples from encryptions of fixed dummy plaintexts, except with negligible probability. This means that the ciphertexts generated from different inputs do not reveal useful information to distinguish the inputs, i.e., whether $\sigma_2$ and $\sigma_3$ originated from the same user that generated $\sigma_1$. %
        
    \item \textbf{Game 6. Remaining Values: }
    The only remaining non-random visible values are the public index $n$ and the constant success value $1$, and non-advantageous information $\AccessPatternLeakage$. Here, $n$ is only a handle for user state, so it does not reveal the challenge bit. After the previous replacements, $\sigma_2$ and $\sigma_3$ have identical distributions. Hence $(\sigma_1,\sigma_2)$ and $(\sigma_1,\sigma_3)$ are identical in the final hybrid game, and the adversary guesses $b$ with a probability of exactly $1/2 + \negl(\lambda)$, which is the same as random guessing except with negligible probability.

\end{itemize}
\end{proof}